\documentclass[11pt]{article}

\usepackage{multirow}
\usepackage{color}
\usepackage{amsmath, amsthm, amssymb, amsfonts}
\usepackage{fullpage}
\usepackage{xspace}
\usepackage{comment}
\usepackage{dsfont}

\usepackage{algorithm}
\usepackage{algpseudocode}
\usepackage{lineno}

\algtext*{EndWhile}
\algtext*{EndIf}
\algtext*{EndFor}

\usepackage{hyperref}
\hypersetup{colorlinks,allcolors=blue}

\ifx\theorem\undefined
\newtheorem{theorem}{Theorem}[section]
\fi
\ifx\example\undefined
\newtheorem{example}[theorem]{Example}
\fi
\ifx\property\undefined
\newtheorem{property}{Property}[section]
\fi
\ifx\lemma\undefined
\newtheorem{lemma}[theorem]{Lemma}
\fi
\ifx\proposition\undefined
\newtheorem{proposition}[theorem]{Proposition}
\fi
\ifx\remark\undefined
\newtheorem{remark}[theorem]{Remark}
\fi
\ifx\corollary\undefined
\newtheorem{corollary}[theorem]{Corollary}
\fi
\ifx\definition\undefined
\newtheorem{definition}[theorem]{Definition}
\fi
\ifx\conjecture\undefined
\newtheorem{conjecture}[theorem]{Conjecture}
\fi
\ifx\fact\undefined
\newtheorem{fact}[theorem]{Fact}
\fi
\ifx\claim\undefined
\newtheorem{claim}[theorem]{Claim}
\fi
\ifx\assump\undefined
\newtheorem{assump}[theorem]{Assumption}
\fi
\ifx\cond\undefined
\newtheorem{cond}[theorem]{Condition}
\fi

\newcommand{\RR}{\mathbb{R}}
\newcommand{\cS}{\mathcal{S}}

\newcommand{\cF}{\mathcal{F}}

\newcommand{\PP}{\Pr}
\newcommand{\EE}{\mathbb{E}}
\newcommand{\NN}{\mathbb{N}}
\newcommand{\cE}{\mathcal{E}}
\newcommand{\cB}{\mathcal{B}}
\newcommand{\cH}{\mathcal{H}}
\DeclareMathOperator{\poly}{poly}
\DeclareMathOperator{\polylog}{polylog}
\newcommand{\cA}{\mathcal{A}}

\newcommand{\wt}[1]{\widetilde{#1}}
\newcommand{\wh}[1]{\widehat{#1}}
\newcommand{\wb}[1]{\overline{#1}}
\newcommand{\ex}{\EE}

\DeclareMathOperator{\var}{Var}

\DeclareMathOperator{\supp}{supp}

\DeclareMathOperator{\HHdim}{HHdim}
\DeclareMathOperator{\VCdim}{VCdim}

\renewcommand{\epsilon}{\varepsilon}

\newcommand{\zo}{\{0,1\}}
\newcommand{\eqdef}{:=}
\newcommand{\set}[1]{\left\{ #1 \right\}}
\newcommand{\tuple}[1]{{\langle{#1}\rangle}}
\providecommand{\card}[1]{\lvert#1\rvert}
\providecommand{\norm}[1]{\left\lVert#1\right\rVert}
\newcommand{\indic}{\mathds{1}}
\newcommand\indice[1]{\indic_{\{#1\}}}

\newcommand{\cSint}{\cS_{\mathrm{int}}}

\newtheorem{problem}{Problem}

\title{Universal Streaming of Subset Norms%
  \thanks{Part of this work was done while the authors were visiting the Simons Institute for the Theory of Computing.}
}

\author{
  Vladimir Braverman%
  \thanks{This material is based upon work supported in part by the National Science Foundation under
Grants No. 1447639, 1650041 and 1652257, Cisco faculty award, and by the ONR Award N00014-18-1-2364.
Email: \texttt{vova@cs.jhu.edu}
  }
  \\
  Johns Hopkins University
  \and
  Robert Krauthgamer%
  \thanks{Work partially supported by ONR Award N00014-18-1-2364,
    the Israel Science Foundation grant \#1086/18, 
    and a Minerva foundation grant from the Federal German Ministry for Education and Research. 
    Email: \texttt{robert.krauthgamer@weizmann.ac.il}
  }
  \\
  Weizmann Institute of Science
  \and
  Lin F. Yang%
  \thanks{Email: \texttt{linyang@ee.ucla.edu}
  }
  \\
University of California, Los Angeles
}

\date{\today}

\begin{document}
\maketitle

\begin{abstract}

Most known algorithms in the streaming model of computation aim to approximate a single function such as an $\ell_p$-norm.
In 2009, Nelson [\url{https://sublinear.info}, Open Problem 30]
asked if it is possible to design \emph{universal algorithms},
that simultaneously approximate multiple functions of the stream. 
In this paper we answer the question of Nelson for the class of
\emph{subset-$\ell_0$} in the insertion-only frequency-vector model.
Given a family of subsets $\cS\subset 2^{[n]}$, we provide a single streaming algorithm that can $(1\pm \epsilon)$-approximate the subset-norm
for every $S\in\cS$. Here, the subset-$\ell_p$ of $v\in \RR^n$ with respect to set $S\subseteq [n]$ is the $\ell_p$-norm of $v_{|S}$
(the vector $v$ restricted to $S$ by zeroing all other coordinates). 

Our main result is a near-tight characterization of the space complexity
of every family $\cS\subset 2^{[n]}$ of subset-$\ell_0$'s in insertion-only streams,
expressed in terms of the ``heavy-hitter dimension'' of $\cS$,
a new combinatorial quantity related to the VC-dimension of $\cS$.
We also show that the more general turnstile and sliding-window models 
require a much larger space usage.  
All these results easily extend to $\ell_1$-norm. 

In addition, we design algorithms for two other subset-$\ell_p$ variants.
These can be compared to the famous Priority Sampling algorithm 
of Duffield, Lund and Thorup [JACM 2007],   
which achieves additive approximation $\epsilon\norm{v}_1$ 
for all possible subsets ($\cS=2^{[n]}$)
in the entry-wise update model.  
One of our algorithms extends their algorithm to handle turnstile updates, 
and another one achieves multiplicative approximation given a family $\cS$.


\end{abstract}

\clearpage

\section{Introduction}

The streaming model of computation, where a space-bounded algorithm
makes only a single pass over an input stream,
has gained popularity for its theoretical significance
and usefulness in practice.
Researchers have designed efficient streaming algorithms for many fundamental problems, including, for example:
moments or norms of a frequency vector $v\in\RR^n$
formed by a stream of additive updates \cite{ams99, Indyk06, iw05};
clustering a stream of points in $\RR^d$ \cite{hm04};
and graph statistics for streams of edge updates \cite{fkmsz05}.

Most algorithms designed for this model solve only a single problem.
For instance, in the extensively studied area of streaming $\ell_p$ norms
of a frequency vector, an algorithm usually makes a pass over the stream,
and then it can use the summary it stores to compute
only one particular norm -- the one it was designed for.
Designing a new algorithm for each statistic can be impractical in some applications.
For example, in network monitoring, it is often desirable to maintain a single summary
of the observed traffic and use this summary for multiple tasks such as approximating the entropy, finding elephant flows and heavy hitters, and detecting DDOS attacks~\cite{liu16,Sekar:2010:RCM:1879141.1879186}.


The importance of multi-functional summaries has been observed in the theory community as well.
Nelson~\cite[Open Problem 30]{sublinear30} asked in 2009 if it is possible to design \emph{universal} algorithms for families of functions.
More formally, given a family $\cF$ of functions of the form $f:\RR^n\to\RR$,
the goal is to compute, in one pass over a stream representing $v\in\RR^n$,
a single summary that can be used to evaluate $f(v)$ for every function $f\in\cF$.
Several algorithms~\cite{bo10a, braverman_et_al:LIPIcs:2015:5324, braverman_et_al:LIPIcs:2015:5325, bbcky17, bcwy16}
provide universal sketches for some families of functions,
for example all symmetric norms in a certain class \cite{bbcky17}.
However, universal algorithms are an exception rather than the rule,
and Nelson's question is still open in any reasonable generality.

A simple systematic method to generate a family $\cF$
from a single function $f:\RR^n\to\RR$ 
is to apply this $f$ to different subsets of coordinates.
More precisely, for every subset $S\subset [n]$
define the function $f_{S}: v \mapsto f(v_{|S})$,
where $v_{|S}$ denotes zeroing out the coordinates of $v$ not in $S$.
In this way, every set system $\cS\subset 2^{[n]}$
describes a family of functions $\{f_S: S\in \cS\}$.
We focus on the basic case where $f$ is an $\ell_p$-norm,
and call such function families \emph{subset $\ell_p$-norms}.
As usual, we wish to approximate each function \emph{multiplicatively},
say within factor $1\pm \epsilon$ for a given $\epsilon \in (0,1)$;
this is clearly stronger than approximating additively by $\epsilon\norm{v}_p$. 

Subset $\ell_p$-norms arise naturally in applications;
for instance, $\cS$ could represent supported queries to a database.
Indeed, the well-known \textsc{Subset Sum} problem~\cite{adlt05, dlt07, s06}
and its variant called \textsc{Disaggregated Subset Sum}~\cite{CDKLT14,t18}, 
are equivalent to our subset $\ell_1$-norm problem
in the entry-wise and insertion-only models, respectively. 
Network-monitoring tasks, such as worm detection, rely on
subset $\ell_p$-norms to approximate \emph{flow} statistics~\cite{dlt07}.
Recall that a network flow is a subset of network traffic defined by a source address, a customer, an organization, an application or an arbitrary predicate. It is folklore that calculating flow volume is simply a subset-$\ell_1$ query, and the number of distinct packets in a flow is a subset-$\ell_0$ query.
Recent work~\cite{Sonata18,Narayana:2017} reiterated that approximating
these subset-$\ell_p$ queries and more general filters
is still an important open problem in network telemetry. 
In another recent example, Ting~\cite{t18} argued that
\textsc{Disaggregated Subset Sum} is widely applicable in ad prediction, 
where future user behaviour is inferred from historical aggregate queries
that have a form of subset-$\ell_1$. 
In both examples, data collection is challenging since future queries can be arbitrary, and thus it is critical to answer large classes of subset-$\ell_1$ queries.
We refer the reader to~\cite[Section 1.4.1]{dlt07} and~\cite[Section 2.2]{t18}
for detailed discussions of these and additional applications in machine learning~\cite{shrivastava2016time},
database query estimation~\cite{vengerov2015join}
and denial of service attacks~\cite{sekar2006lads}.

Our \emph{main contribution} to universal streaming is a \emph{near-tight}  characterization, {for every $\cS\subset 2^{[n]}$},
of the space complexity of subset-$\ell_0$'s in insertion-only streams.
We stress that this problem asks to count the distinct items 
(non-zero coordinates of $v$) inside every subset $S\in\cS$. 
Our characterization connects the space complexity for a set system $\cS$
to a combinatorial notion that we call the \emph{heavy-hitter dimension} of $\cS$, 
which counts the maximum possible number of coordinates in a single $v\in\RR^n$
that may be a ``heavy hitter'' for some $S\in\cS$
(see Definition~\ref{defn:HHdim} for full details). 
%
This notion is related to VC dimension by $\VCdim(\cS) \le \HHdim(\cS)$,
however the gap between the two is not bounded by any fixed factor. 
We in fact prove the above inequality 
and make use of it in Section~\ref{sec:SubsetLpForAll}. 
Throughout, $\wt{O}(\cdot)$ suppresses a $\polylog(n)$ factor,
and $O_\epsilon(\cdot)$ suppresses a factor depending only on $\epsilon$;
taken together, $\wt{O}_\epsilon(f)$
stands for $O(g(\epsilon)(\log^{O(1)} n))\cdot f$ for some function $g$.

\begin{theorem}[Informal Statement of Theorems~\ref{thm:alg-main} and \ref{thm:lower-bound-main}]
For every $\cS\subset 2^{[n]}$ and $\epsilon\in(0,1)$,
there is a randomized universal algorithm for insertion-only streams,
that makes one pass using $\wt{O}_\epsilon(\HHdim(\cS))$ words of storage,
and can then $(1+\epsilon)$-approximate each subset $\ell_0$-norm from $\cS$
with high probability.
Moreover, every such algorithm requires $\Omega(\HHdim(\cS))$ bits of storage.
\end{theorem}

To illustrate the scope of this result, we present in Table~\ref{tab:hhdim}
a few examples and properties of the heavy-hitter dimension
(their proofs appear in Section~\ref{sec:HHexamples}). 
One interesting example is the family of all large intervals in $[n]$,
namely, of size $\Omega(n)$, which has dimension $O(1)$.
A second one is a family of $\poly(n)$-many uniformly-random sets
(every set contains every index with probability $1/2$),
which has dimension $O(\log n)$ with high probability.
For both examples, our algorithm uses small space (polylogarithmic in $n$)
to achieve multiplicative $(1+\epsilon)$-approximation,  
which was not known before. 
(For intervals, additive approximation $\epsilon \norm{v}_1$
can be achieved by several known algorithms, including quantile estimates, 
range counting, and heavy-hitters over dyadic intervals.)

Another interesting example is a two-dimensional family
derived from the last property in the table, as follows. 
Suppose each index $i\in[n]$ is actually a pair $(i_1,i_2)\in [n_1]\times[n_2]$,
for instance the source and destination address of a packet. 
Let $\cS_1$ be the family of all large intervals with respect to $i_1$,%
\footnote{Strictly speaking, we cannot simply ignore $i_2$, 
  but we can order the $n=n_1 n_2$ pairs according to $i_1$ 
  and take all large intervals under this ordering.
  This argument exploits the fact that the heavy-hitter dimension
  is permutation-invariant, i.e., not affected by reordering the coordinates. 
}
and similarly $\cS_2$ with respect to $i_2$. 
Each of $\cS_1,\cS_2$ has dimension $O(1)$, and thus also their union-product, 
which implies that our algorithm will use small space (polylogarithmic in $n$)
can estimate distinct elements in subsets of form
$i_1\in[a,b] \vee i_2\in[c,d]$, 
for instance the logical-or between a range of source addresses 
and a range of destination addresses.

\begin{table}[htb!]
\centering
\begin{tabular}{ccc}
\hline\hline
Set System & $\HHdim$ & Description\\
\hline
$\set{ S_1,\ldots,S_k }$ & $\le k$ &  any sets (tight for disjoint sets)\\
$\set {S\subset[n] :\ \card{S} \ge n-k }$ & $k+1$ & sets missing few coordinates\\
$\set{[i..i']:\ i'-i+1 \ge k }$  & $\Theta(n/k)$ & intervals of size $\ge k$ \\
$\set{ s_1,\ldots,s_k: \mbox{ all } s_{i,j}\sim \cB(p) }$ & whp $O(\log(nk)/q)$ & $k$ random subsets of density $q$\\
\hline
$\set{ S'\cup S'':\ S',S'' \in \cS } $& $\HHdim(\cS)$ & self union of $\cS$ \\
$\cS_1\cup \cS_2 $& $\le \HHdim(\cS_1) + \HHdim(\cS_2)$ & sub-additivity\\
$\set{ S_1\cup S_2:\ S_1\in\cS_1, S_2\in \cS_2 } $ & $\le \HHdim(\cS_1) + \HHdim(\cS_2)$ & union-product of $\cS_1,\cS_2$\\
\hline
\end{tabular}
\caption{Simple examples and basic properties of heavy-hitter dimension over domain $[n]$. 
  \label{tab:hhdim}
}
\end{table}

Let us consider some natural extensions of the above theorem.
First, our algorithm extends to subset-$\ell_1$ as well,
as shown in 
Theorem~\ref{thm:alg-l1}. 
Second, it is stated for insertion-only streams,
however for the more general turnstile and sliding-window models
(i.e., streams with both insertions and deletions or streams where old items expire),
we show that the subset-$\ell_p$ problem, for any $p\ge 0$, 
requires space $\Omega(n)$ even if $\HHdim(\cS) = O(1)$.
This is striking because such a large separation between insertion-only and turnstile stream or sliding window algorithms is rare.
Indeed, it may be instructive to see why the smooth-histograms technique
of~\cite{BO07} fails in this case.

\begin{theorem}[Informal Statement of Theorems~\ref{thm:turnstile-lower-bound} and \ref{thm:sliding-lower-bound}]
There exists $\cS\subset 2^{[n]}$ with $\HHdim(\cS) = O(1)$,
such that every universal streaming algorithm
achieving multiplicative approximation for subset-$\ell_p$ for $\cS$
requires $\Omega(n)$ bits of space
in both the turnstile and sliding-window models.
\end{theorem}

\paragraph{Variants of the problem.}
Duffield, Lund and Thorup~\cite{dlt07} consider a similar problem in the  entry-wise update model,
in which each entry of the vector appears in the stream at most once.
Their ``subset-sum'' problem is equivalent to our subset-$\ell_1$ problem
if all entries of the vector are non-negative.%
\footnote{The non-negativity is a mild assumption
  in this entry-wise update model,
  because one can easily separate the positive and negative entries
  and execute in parallel two algorithms.
}
They devise a Priority Sampling algorithm that approximates
the subset-$\ell_1$ of every subset $S\subset [n]$,
achieving in fact an optimal space usage for this model \cite{s06}.
However, their result actually guarantees an \emph{additive approximation},
i.e., the error for every subset-$\ell_1$ query
is proportional to the $\ell_1$-norm of the entire vector.%
\footnote{Indeed, Theorem~1 of~\cite{s06} bounds the variance of the estimator by $\|v\|_1^2/(k-1)$, where $k$ is number of samples being stored.
  This implies that with high probability, the estimator's additive error
  is at most $O(\|v\|_1/\sqrt{k-1})$.
}
We additionally provide two extensions to their results in new directions.

The first extension is a full characterization of space complexity of \emph{multiplicative approximation} of subset-$\ell_p$'s
in the entry-wise update model.
In contrast to the results of \cite{dlt07,s06},
once a multiplicative approximation is required,
the space complexity depends on the query set system $\cS$.
Indeed, by modifying the priority sampling algorithm of~\cite{dlt07,s06}
and employing our lower bound,
we show (in Theorem~\ref{thm:entry-wise-model}) that the space complexity
is now precisely $\wt{\Theta}(\HHdim(\cS))$.

Our second extension achieves the same \emph{additive approximation},
but in the more general \emph{turnstile model}
(i.e., additive updates to entries).
Similarly to the algorithm of~\cite{dlt07,s06},
our algorithm for the subset-$\ell_p$ problem
achieves additive error $\epsilon \|v\|_p$
with space complexity that does not depend on the query set system $\cS$.
This result is summarized in the following theorem;
we note that a matching lower bound follows immediately from known results,
as the case $S=[n]$ is the usual approximation of $\ell_p$ norm.

\begin{theorem}[Informal Statement of Theorem~\ref{thm:additive}]
There exists a one-pass streaming algorithm that,
given a stream of additive updates to a vector $v\in \RR^n$,
uses only $\wt{O}(1)$ words of space for $0\le p\le 2$
(and $\wt{O}_{\epsilon}(n^{1-2/p})$ words for $p>2$),
and can then approximate $\|v_{|S}\|_{p}$ within additive error
$\epsilon \|v\|_p$ for each $S\subset [n]$ with high probability.
\end{theorem}

\paragraph{Summary.}
Our results are summarized in Table~\ref{tab:sum-res}, 
where the first row lists the main results.

\begin{table}[htb!]
	\renewcommand{\arraystretch}{1.2}%
	\centering
	\begin{tabular*}{\textwidth}{c @{\extracolsep{\fill}}cccc}
		\hline\hline
		Problem & Update Model & Approximation & Space & Theorems\\
		\hline
		subset-$\ell_0$ or $\ell_1$& insertion-only & multiplicative &  		$\wt{\Theta}(\HHdim(\cS))$  &\ref{thm:alg-main}, \ref{thm:lower-bound-main}, \ref{thm:alg-l1}\\ 
		\hline
		\multirow{3}{*}{subset-$\ell_p$}& turnstile & \multirow{ 3}{*}{multiplicative} &  $\Omega(n)$  & \ref{thm:turnstile-lower-bound}\\ 
		& sliding-window &  &  $\Omega(n)$  & \ref{thm:sliding-lower-bound}\\ 
		& entry-wise &  &  $\wt{\Theta}(\HHdim(\cS))$ & \ref{thm:entry-wise-model}\\
		\hline
		subset-$\ell_p$ & turnstile& additive & \begin{tabular}{lc}
		$\wt{\Theta}(1)$ & for $0\le p\le 2$ \\
		$\wt{\Theta}(n^{1-2/p})$&  for $p>2$
		\end{tabular} & \ref{thm:additive}\\
		\hline
	\end{tabular*}
	\caption{Summary of our results 
          \label{tab:sum-res}}
\end{table}

\subsection{Related Work}
\label{sec:related}
There is a large body of works that deals with approximating functions of a vector, i.e., norms and heavy hitters, in the streaming model of computation.
For instance \cite{ams99} is the first paper that systematically studies the $\ell_p$ norm approximation of a streaming vector;
\cite{Indyk06} gives the first near-optimal algorithm (in terms of $n$) for $\ell_p$ for all $0\le p\le 2$;
\cite{iw05} gives the first near-optimal algorithm for $\ell_p$ for all $p>2$;
\cite{iw03, w04, cks03, bjks04} give tight lower bounds on this problem with respect to the approximation parameter $\epsilon$ and dimension $n$.
There is a sequence of papers gradually improving the space complexity with respect to other parameters and studying variants of the problem.
Due to the lack of space we mention only a small subset of relevant papers 
\cite{CCF04, cks03, bgks06, knw10a, knw10b, knpw11, ako11, andoni2017high, bksv14, bcinww16,bcwy16};
and references therein. 
Most of these papers design methods that approximate a single function such as an $\ell_p$ norm for a fixed $p$.

Our setting is also related to the ``subset sum'' problem \cite{adlt05, dlt07, s06, t18} where one is interested in approximating 
a sum of the entries of a vector indexed by a subset.
It is not difficult to see that our problem is the same as the objective in \cite{dlt07} when the input is restricted to non-negative vectors; indeed the subset-sum problem is equivalent to the subset-$\ell_1$ problem. However, the model in \cite{dlt07} is slightly different.
In \cite{dlt07} the algorithm sees each coordinate of the frequency vector at most once. In this paper we consider the additive updates streaming model that allows incremental updates to the coordinates of the frequency vector. Thus, our model generalizes the model in \cite{dlt07}.
In addition, the algorithms in \cite{dlt07} solve the subset-sum problem but with an additive error.




\subsection{Preliminaries}
\label{sec:prelims}
We identify a binary vector $s\in\{0,1\}^n$ as a subset in $[n]$.
For two vectors $u, v\in \RR^{n}$, we denote $u\circ v\in \RR^n$ as the Hadamard product, i.e., each $(u\circ v)_i = u_i v_i$.
We denote the support of $v$, $\supp(v)\subset [n]$, as the set of non-zero coordinates in $v$, i.e., $\supp(v) =\{i\in [n]: v_i\neq 0\}$.
For each $p>0$, we denote the $\ell_p$ norm of a vector $v\in \RR^n$ as $\|v\|_p = (\sum_{i\in [n]}|v_i|^p)^{1/p}$.
For $p=0$, $\|v\|_0\eqdef|\supp(v)|$ is the size of the support of $v$.
For $p=\infty$, $\|v\|_\infty\eqdef\max_{i\in [n]}|v_i|$.
Note that for $p< 1$, $\ell_p$ is not a ``norm'' but were called a norm by convention.

In this paper we are focusing on the updates of a vector $v\in \RR^n$.
In the \emph{insertion-only} model, the input is a stream $\tuple{a_1, \ldots, a_m}$, 
where each item $a_j\in [n]$ represents an increment to coordinate $a_j$ 
of a vector $v\in\RR^n$, which is initialized to all zeros. 
Thus the accumulated vector is $v=\sum_{j=1}^m e_{a_j}$, 
where $\set{e_i:i\in[n]}$ is the standard basis. 
Here $m$ is usually assumed to be upper bounded by $\poly(n)$.
In the \emph{turnstile} model, the input is a stream $\tuple{(a_1, \Delta_1), \ldots, (a_m, \Delta_m)}$, where each item $(a_j, \Delta_j)\in [n]\times\{-1, 1\}$
represents an increment to coordinate $a_j$ 
of a vector $v\in\RR^n$ by $\Delta_j$.%
\footnote{A more general model allows $\Delta_i \in\{-M, \ldots, M\}$
for some $M=\poly(n)$. The space/time usage of the $M=1$ case is only up to $O(\log(n))$ factor worse. We use $M=1$ for sake of representation.
}
Thus the accumulated vector is $v=\sum_{j=1}^m \Delta_j\cdot e_{a_j}$.

We are interested in the following problem.
\begin{problem}[Subset-$\ell_p$]
	Let $\alpha\ge 1$ be a parameter. 
	Given a set of binary vectors $\cS\subset \{0,1\}^n$, design a one-pass algorithm over a stream of updates to vector $v\in\RR^n$, such that, at the end of the stream, the algorithm outputs a function $E:\cS\rightarrow \RR$ that satisfies,  
	\[
	\forall s\in \cS, \qquad \Pr[\|v\circ s\|_p\le E(s) \le \alpha\|v\circ s\|_p] \ge 0.9.
	\]
	We call this problem the \emph{$\alpha$-approximation subset-$\ell_p$} problem w.r.t. $\cS$.
\end{problem}
Note that the above definition requires the algorithm to approximate, for each given $s\in \cS$, the $\|v\circ s\|_p$ well.
A standard parallel repeating argument can lead to the ``for-all'' guarantee, i.e., the algorithm succeeds on approximating $\|v\circ s\|_p$ for all $s\in \cS$.
However, we pay an additional $\log|\cS|$ factor in the space -- this can be linear in $n$ if $|\cS|$ is large.
It would be interesting if one can design a for-all algorithm with space not depending $\log|\cS|$.

Note that the set system $\cS$ is given to the algorithm via a read-only tape, hence the space of storing $\cS$ is not counted.
A variant of this problem is the additive approximation problem.

\begin{problem}[Additive Subset-$\ell_p$]
	For a set of binary vectors $\cS\subset \{0,1\}^n$, design a one-pass algorithm over a stream of updates to some underlying vector $v\in\RR^n$ such that after one pass over the stream, the algorithm outputs a function $E:\cS\rightarrow \RR$ satisfies, 
	\[
	\forall s\in \cS, \quad \Pr[\|v\circ s\|_p - E(s)\big|\le 
	\epsilon \|v\|_p]\ge 0.9.
	\]
	We call this problem the \emph{$\epsilon$-additive-approximation subset-$\ell_p$} problem.
\end{problem}


\subsection{Technical Overview}
\label{sec:techniques}
\paragraph{Multiplicative Subset-$\ell_{p}$ Algorithm for $p\in\{0,1\}$. }

Estimating $\ell_0$ of a stream is a well-studied problem, 
for instance the first streaming algorithm was given in~\cite{fm85},
and the problem's space complexity was settled in~\cite{knw10a}. 
Let us first recall a classical sample-and-estimate technique for this problem
(see e.g.~\cite{bjkst02}).
Here, the algorithm subsamples each coordinate of $v$ with some probability $p$,
and then uses the sampled non-zero coordinates to estimate $\|v\|_0$
(simply count their number and divide by $p$). 
Suppose we could guess the correct rate $p$, 
such that number of non-zero samples is about $\Theta(\epsilon^{-2})$; 
then we would obtain a good estimate to the $\ell_0$ of the stream,
i.e., a $(1\pm \epsilon)$-approximation with constant probability. 
The number of guesses is at most $\Theta(\log n)$, since $\|v\|_0\le n$,
and we can try all of them in parallel. 
Observe that this algorithm actually stores all distinct samples up to a point
-- when the samples for a guess $p$ exceeds the $O(\epsilon^{-2})$ space bound,
the algorithm starts rejecting any extra samples. 

Consider now approximating $\|v\circ s\|_0$ for any $s$
in a known set system $\cS\subset 2^{[n]}$. 
To use the above sample-and-estimate technique, 
the guess $p$ should be chosen according to $\|v\circ s\|_0$.
However, an algorithm that is not tailored to $s$
will store (distinct) samples from all $\supp(v)$,
and thus it might reach its $O(\epsilon^{-2})$ space bound and start rejecting samples,
without storing enough samples from $\supp(v\circ s) \subseteq \supp(v)$.
The challenge is thus to store enough samples from $\supp(v\circ s)$
for every $s\in\cS$. 
Our idea is to rely on the structure of the set system $\cS$,
and store every sample that might be necessary for any $s\in\cS$,
which clearly maintains the correctness (accuracy guarantee).
However, this might require large space, perhaps even linear in $\card{\cS}$,
and our solution is to \emph{actively delete} samples that are not necessary. 

To formalize this idea, we let the algorithm store a set $\cH \subset [n]$
of (distinct) samples from the stream.
Now whenever the number of samples from some $s\circ v$
is smaller than our $O(\epsilon^{-2})$ bound, all these samples are stored,
and then $\cH$ can always be used to estimate $\|s\circ v\|_0$. 
However, when some $i\in\cH$ is no longer necessary for any $s\in\cS$ 
(which might happen as new samples are stored), 
the algorithm deletes this $i$ from $\cH$. 
The question is then: what is the maximum possible size of $\cH$? 
Luckily, we can show that $|\cH|= O[ \epsilon^{-2}\cdot \HHdim(\cS)]$
via an inductive argument,
whose base case is precisely the heavy-hitter dimension. 
Maintaining $\cH$ in a streaming fashion is straightforward 
and requires only $O[\epsilon^{-2}\cdot \HHdim(\cS)]$ words of space.
Recalling there are $O(\log n)$ guesses for $p$,
the algorithm actually stores $O(\log n)$ sets of samples, 
which altogether can simulate the sample-and-estimate algorithm
for any $s\in\cS$ given at the query phase, 
to achieve a multiplicative approximation of $v\circ s$. 
This result is presented in Theorem~\ref{thm:alg-main}.

In insertion-only streams, the $\ell_1$ norm is just the sum of each coordinate.
We can thus reduce the $\ell_1$ estimation problem to a new vector space of dimension $nm$, where $m$ is the length of the stream. 
We show that the converted set system has exactly the same heavy-hitter dimension, yielding again an algorithm with space usage $O[\epsilon^{-2}\cdot \HHdim(\cS)]$.
This result is presented in Theorem~\ref{thm:alg-l1}.

The upper bound for subset-$\ell_p$ norm in the entry-wise update model follows similar ideas to store a small subset that is important to the set system.
The only difference is that we use the priority sampling technique \cite{dlt07,s06} as the bottom-level algorithm.
This result is presented in Theorem~\ref{thm:entry-wise-model}.

\paragraph{Lower Bound for Subset-$\ell_p$.}
Our lower bound is via reduction from the INDEX problem. 
Suppose we have a set system with heavy hitter dimension $\HHdim(\cS)$, we can then find a vector $v$ with $\HHdim(\cS)$ non-zero coordinates and for each coordinate $v_i$, there exists an $s^{i}\in S$ such that $\{i\}=\supp(s^{i}\circ v)$.
Therefore, we can encode an INDEX instance into the non-zero coordinates of $v$ and by approximating $\|v\circ s^{i}\|_{p}$ multiplicatively for any $p$, we can have a protocol for the INDEX problem. 
This implies an $\Omega(\HHdim(\cS))$ lower bound.
This result is presented in Theorem~\ref{thm:lower-bound-main}.

\paragraph{Strong Lower Bound in the Turnstile Model and Sliding Window Model.}
It is striking that in the turnstile model or sliding window model, there does not exists sub-linear one-pass  multiplicative approximation subset-$\ell_p$ algorithms even for a very simple set system. 
We show that for a simple set system, e.g, a set system contains all the intervals with size $n/2$, which has heavy hitter dimension $O(1)$, any multiplicative approximation of subset-$\ell_p$ for any $p\ge 0$ requires $\Omega(n)$ space.
We show this via a reduction from the Augmented INDEX problem. 
In this problem, Alice has a binary vector $x\in \RR^n$ Bob has an index $j\in [n]$ and $x_{j+1}, x_{j+2}, \ldots, x_{n}$.
Alice sends one round of message to Bob and Bob needs to determine what is $x_j$.
It has been shown in \cite{BJKK04b} that, any constant-probability success protocol for this problem requires $\Omega(n)$ bits of space.
We construct a protocol using the subset-$\ell_p$ algorithm.
Alice simply maps each of its coordinates of $x$ to some stream updates. 
Bob removes all $x_{j'}$ for $j'> j$.
Bob then picks the interval that contains at most one non-zero coordinate -- $x_j$ -- and asks the algorithm to compute the $\ell_p$ norm.
Hence any multiplicative approximations can be used to decide whether $x_j$ is $0$.
Thus, any algorithm in the turnstile model requires $\Omega(n)$ space for this simple set system.
Similar lower bounds can be shown for the sliding window model with a lower bound of $\Omega(\min(W, n))$, where $W$ is the window size.
These results are formally presented in Theorem~\ref{thm:turnstile-lower-bound} and Theorem~\ref{thm:sliding-lower-bound}.

\paragraph{Additive Approximation. }
Our additive approximation to the subset-$\ell_p$ norm follows a similar flavor of the priority sampling algorithm \cite{dlt07, s06}.
We use the algorithmic idea appeared in \cite{bvwy18} (similar ideas also appear earlier in \cite{ako11} and \cite{andoni2017high}, but of different form).
To approximate the $\ell_p$ norm of a vector, we first generate $n$ pseudo-randomized random numbers to scale each entry of the input vector $v$, which can be implemented using $\Theta(\log n)$ space in the streaming setting.
If the distribution of the random numbers has a nice tail, e.g, $\Pr[X>x] = 1/x^p$, the $\ell_2$-heavy hitter of the scaled vector can be shown to be a good estimation of the $\ell_p$ norm.
The scaling is ``oblivious'' to the subset, i.e., for each $s\in 2^{[n]}$, the $\ell_2$-heavy hitter of $s\circ v'$ is a good estimator to $\|s\circ v\|_p$, where $v'$ is the scaled version $v'$. 
This result is presented in Theorem~\ref{thm:additive}.


\section{The Streaming Complexity of Subset-$\ell_p$}
\label{sec:SubsetLp}

In this section we study algorithms for the Subset-$\ell_p$ problem
(namely, achieve multiplicative approximation) for $p=0,1$. 
Our main finding is that the space complexity in insertion-only streams
is characterized by the following combinatorial quantity.

\begin{definition}[Heavy-Hitter Dimension]
  \label{defn:HHdim}
For a set system $\cS\subset 2^{[n]}$ and a vector $v\in \RR^n$, we denote the 
$H(\cS,v)$ as the set of heavy hitters induced by $\cS$:
\[
H(\cS,v) = \big\{i\in[n]:\ \exists s\in\cS \mbox{ s.t. } \supp(s\circ v) = \{i\}\big\}.
\]
and define the \emph{heavy-hitter dimension} of $\cS\subset\zo^n$ as 
\[
  \HHdim(\cS)\eqdef \sup_{v\in\RR^n} |H(\cS, v)|. 
\]
\end{definition}

Our main result is a streaming algorithm 
with space complexity that is linear in the heavy-hitter dimension, 
i.e., $\wt{O}_\epsilon(\HHdim(\cS))$,
see Theorem~\ref{thm:alg-main} in Section~\ref{sec:SubsetLpALg}.
We then provide several complementary results.
From the direction of space lower bounds, 
we prove a linear lower bound $\Omega(\HHdim(\cS))$,
which matches our algorithm above (in Section~\ref{sec:SubsetLpLB}),
and also a much bigger bound for turnstile and sliding-window streams, 
which separates these richer models from insertion-only streams
(in Section~\ref{sec:SubsetLpTurnstile}).
From the direction of applications of our algorithmic techniques, 
we extend our algorithm to the ``for-all'' guarantee
(in Section~\ref{sec:SubsetLpForAll}),
and to the case $p=1$
(in Section~\ref{sec:SubsetLpL1}), 
and we also design a variant for the more restricted model of 
entry-wise updates
(in Section~\ref{sec:SubsetLpEntrywise}).

\subsection{Examples and Properties of Heavy-Hitter Dimension}
\label{sec:HHexamples}

We present a few simple examples and basic properties of
the heavy-hitter dimension that may be useful in applications,
essentially proving the bounds shown in Table~\ref{tab:hhdim}. 

We begin with an alternative description of $\HHdim(\cS)$
where we view the set system $\cS\subset 2^{[n]}$ as an \emph{incidence matrix},
i.e., a $0-1$ matrix describing the incidence between sets $S\in\cS$
and coordinates $i\in [n]$.
Recall that a matrix $M\in\zo^{k\times k}$ is called a \emph{permutation matrix}
if every row and every column contain a single non-zero, i.e., exactly one $1$.
Clearly, up to reordering the rows and/or columns,  
such a matrix can be viewed as an identity matrix.

\begin{lemma}[Permutation Submatrix]
  \label{lem:HHdimPermuation}
Let $\cS\subset 2^{[n]}$,
and let $M\in\zo^{|\cS|\times n}$ be its incidence matrix. 
Then $\HHdim(\cS)$ is exactly the maximum size (number of rows/columns)
in a submatrix of $M$ that is a permutation matrix. 
\end{lemma}

\begin{proof}
Denote $\cS=\set{s_1,s_2,\ldots}$,
and let $k$ be the largest size of permutation submatrix of $M$.
Suppose that this submatrix is formed
by rows $i_1,\ldots,i_k$ and columns $j_1,\ldots,j_k$. 
Consider a vector $v\in\zo^n$ with $1$ exactly in coordinates $j_1,\ldots,j_k$.
Then it is straightforward to see that 
\[
  \forall l\in [k],
  \quad
  \supp(s_{i_l}\circ v) =\{j_l\}. 
\]
Thus $k\le \HHdim(\cS)$.

For the other direction, 
let $u\in \RR^{n}$ be a vector that realizes $\HHdim(\cS)$.
Then there are sets $s_{i_1}, \ldots, s_{i_{\HHdim(\cS)}}\in \cS$
and coordinates $j_1, \ldots, j_{\HHdim(\cS)}\in [n]$ such that
\[
  \forall l\in [\HHdim(\cS)],
  \quad
  \supp(s_{i_l}\circ u) = \{j_l\}. 
\]
It is easily verified that rows $i_1,\ldots, i_{\HHdim(\cS)}$
and columns $j_1, \ldots, j_{\HHdim(\cS)}$ form
a permutation submatrix of $M$. 
Thus, $\HHdim(\cS)\le k$, which completes the proof.
\end{proof}

Using the above lemma one can analyze the heavy-hitter dimension
of several explicit set systems,
as listed in the first four lines in Table~\ref{tab:hhdim}.
The proofs are straightforward we give only one for example
in Proposition~\ref{prop:randomsets} below. 
This lemma also implies bounds for several operations on set systems,
as listed in the last three lines in Table~\ref{tab:hhdim}.
The proofs are straightforward, and we give only one for example
in Proposition~\ref{prop:subadd1} below. 

\begin{proposition}[Random Sets]
  \label{prop:randomsets}
Let $\cS\subset 2^{[n]}$ be a set system of size $\card{\cS}=k$,
whose incidence matrix $M$ is formed entries that are
independent Bernoulli random variables with parameter $p\in(0,1/2]$. 
In other words, every set $S_i\in\cS$ contains every coordinate $j\in[n]$ independently with probability $p$. 
Then
\[
  \Pr\big[ \HHdim(\cS) \le O(\log(nk)/p) \big]
  \geq
  1-1/\poly(nk). 
\] 
\end{proposition}

\begin{proof}
Fix $t$ rows and $t=c\log(nk)/p$ columns of $M$
and consider the corresponding submatrix $M'$. 
The probability that $M'$ is an identity matrix is 
$p^{t}(1-p)^{t^2-t} \le e^{-pt^2/2}$.
For the event $\HHdim(\cS) \ge t$ to occur
there must be ordered sequences of $t$ rows and $t$ columns
that yield an identity $M'$. 
Since the number of such choices is at most $n^t\cdot k^t$, we obtain 
\[
  \Pr\big[\HHdim(\cS) \ge t\big]
  \le (nk)^t\cdot e^{-pt^2/2}
  \le e^{-pt^2/4}
  \le 1/\poly(nk). 
\]
\end{proof}

\begin{proposition}[Sub-Additivity]
  \label{prop:subadd1}
For every $\cS_1, \cS_2\subset2^{[n]}$, 
\[
  \HHdim(\cS_1\cup \cS_2) \le \HHdim(\cS_1) + \HHdim(\cS_2).
\]
\end{proposition}

\begin{proof}
Let $M_1$ and $M_2$ be the incidence matrices of $\cS_1$ and $\cS_2$,
respectively. 
Then the incidence matrix of $\cS_1 \cup \cS_2$ is,
using block-matrix notation, simply
$M =  \bigl[\begin{smallmatrix} M_1 \\ M_2 \end{smallmatrix}\bigl]$. 
Every permutation submatrix of $M$ 
can be partitioned into $M_1$ and $M_2$.
As each part must contain a permutation submatrix that uses all its rows,
the proposition follows.
\end{proof}

\subsection{Streaming Algorithm for Subset-$\ell_0$}
\label{sec:SubsetLpALg} 

We now design a one-pass streaming algorithm for the Subset-$\ell_0$ problem 
in insertion-only stream. 
Recall that in this model the input is a stream $\tuple{a_1, \ldots, a_m}$, 
where each item $a_j\in [n]$ represents an increment to coordinate $a_j$ 
of a vector $v\in\RR^n$. 
The streaming algorithm has two phases,
an update phase that scans the stream,
and a query phase that evaluates a query $s\in \cS$.
(In one case below, the query phase formally does not require a query $s$,
and reports one list that implicitly represents every query $s\in\cS$.)
We assume that the set system $\cS$ is given to the algorithm via 
a read-only tape, and thus requires no storage.
See Section~\ref{sec:prelims} for detailed definitions. 
 
The algorithm uses a well-known technique of subsampling the coordinates of $v$
(i.e., the set $[n]$) at a predetermined rate $p\in(0,1]$, 
and producing an estimate only if 
the resulting vector has $O(\epsilon^{-2})$ non-zeros. 
Usually, counting the number of sampled non-zeros requires little space,
but this is more challenging in our case of all norms $s\in\cS$. 

The key to bounding the total space usage is the following proposition,
which bounds the global number of samples stored,  
when each of these samples is ``needed'' locally by some subset $s\in \cS$. 
This condition has a parameter $k$, 
and the reader may initially think of $k=1$. 

\begin{proposition} \label{prop:subset}
Let $\cS\subset\zo^n$ be a set system.
Suppose $Z\subset [n]$ and $k\ge 1$ are such that for every $i\in Z$ 
(in words, index $i$ is ``$k$-needed'' by some $s\in \cS$)
\begin{align} \label{eq:kneeded}
  \exists s\in \cS, &\quad  
  i\in s  
  \text{ and }
  \card{Z\cap s} \le k. 
\end{align}
Then $\card{Z} \le k\cdot \HHdim(\cS)$. 
\end{proposition}

\begin{proof}
We proceed by induction on $k$. 
For the base case $k=1$, 
consider a vector $v\in\zo^n$ whose support is exactly the given $Z$.
Then for every $i\in Z$, there is $s\in \cS$ such that $Z\cap s=\set{i}$,   
and thus $D(s,v)=\set{i}$. 
It follows that
$\card{Z} 
  \le \card{H(\cS,v)} 
  \le \HHdim(\cS)$. 

For the inductive step, consider $k\ge 2$. 
Given $Z$, construct $A\subset Z$ as follows. 
Start with $A=Z$ and iteratively remove from it 
an index $i\in A$ if there is no $s\in \cS$ with $A\cap s=\set{i}$ 
(i.e., if $i$ is not $1$-needed), 
until no such index $i$ exists.
We claim that the final set $A$ satisfies 
\[
  \forall s\in \cS,
  \qquad
  Z\cap s \neq\emptyset   \ \Leftrightarrow\ 
  A\cap s \neq\emptyset .
\]
For the forward direction, 
observe that initially $\card{A\cap s} = \card{Z\cap s} \ge 1$,
and that no iteration never decreases any $\card{A\cap s}$ from $1$ to $0$. 
The reverse direction is obvious because $A\subset Z$. 

We can verify that $Z\setminus A$ satisfies the induction hypothesis,
i.e., that every $i\in Z\setminus A$ is $(k-1)$-needed. 
Indeed, since $i\in Z$ it is $k$-needed by some $s\in \cS$,
as expressed by~\eqref{eq:kneeded},
and by the claim, this same $s$ also satisfies $\card{A\cap s}\ge 1$.
Hence, $\card{(Z\setminus A)\cap s} \le k-1$, which shows $i$ is $(k-1)$-needed. 
Applying the induction hypothesis, we have
$\card{Z\setminus A} \leq (k-1)\cdot \HHdim(\cS)$. 

In addition, $A$ satisfies the induction's base case $k=1$,
because the iterations stop when every $i\in A$ is $1$-needed.
Hence $\card{A}\le \HHdim(\cS)$, and we conclude that 
$\card{Z} = \card{Z\setminus A} + \card{A} \le k\cdot \HHdim(\cS)$. 
\end{proof}
 
Our algorithm is based on simulating 
the simple estimator defined in the following lemma.
(The difficulty will be to apply it to $s\circ v$ for $s\in\cS$ 
that is not known in advance.) 

\begin{lemma} \label{lem:qv-pairwise}
Fix $v\in\RR^n$, 
and sample its coordinates to form $v'\in\RR^n$ as follows.
Suppose each coordinate is $v'_i = v_i X_i$, 
where $X_1,\ldots,X_n$ are pairwise-independent identically distributed Bernoulli random variables with parameter $p\in(0,1]$. 
Then 
\begin{align*}
  & \ex\big[ \tfrac{1}{p}\|v'\|_0\big] = \|v\|_0 , 
  \quad
  \var\big( \tfrac{1}{p}\|v'\|_0 \big) = \tfrac{1-p}{p} \|v\|_0, 
  \quad\text{ and }\quad
  \\
  & 
  \Pr\Big[ \big| \tfrac{1}{p}\|v'\|_0 - \|v\|_0 \big| 
           \ge 3 (\tfrac{1-p}{p} \|v\|_0)^{1/2} 
     \Big]
  \le \tfrac{1}{9} .
\end{align*}
\end{lemma}
\begin{proof}
The expectation and variance follow from direct calculation
(note that it suffices to assume pairwise-independence).
The tail bound is straightforward from Chebyshev's inequality. 
\end{proof}

Our basic algorithm for storing a subsample of the coordinates
is described in Algorithm~\ref{alg:single-estimator}. 
The idea is to sample the universe $[n]$ at rate $p\in(0,1]$
using pairwise independent random variables $\xi_1,\ldots,\xi_n\in\zo$ 
(this can be viewed as a hash function $\xi:[n]\to\zo$).
We store the sampled items (coordinates) from the stream 
so long as they are ``needed'' by some $s\in \cS$, 
or more precisely, $U$-needed in the sense of Proposition~\ref{prop:subset}.
Here, the budget parameter $U$ represents 
a ``local'' bound that holds separately for each $s\in\cS$. 
We show that at the end, for every $s\in\cS$, 
if $s$ contains at most $U$ sampled coordinates, 
then all these samples are completely stored;
otherwise, the number of samples stored is at least $U$. 
We will later use this algorithm to simulate an offline pairwise sampling
from a desired $s\in\cS$.

\begin{algorithm}
\caption{Bounded-Sampler for $\cS\subset\zo^n$ \label{alg:single-estimator}}
\begin{algorithmic}[1]
  \State \textbf{Input:} $p\in (0,1], U\in [1,n]$, an insertion-only stream $\langle a_1, \ldots, a_m\rangle$, where each item $a_j\in [n]$ 
  \State \textbf{Initialize:} 
  \State \quad $\cH\gets \emptyset$
  \State \quad pick random $\xi_1,\ldots,\xi_n\in\zo^n$, 
  with each $\PP[\xi_i = 1] = p$ and pairwise independent
  \State \textbf{Update($a_j$):}
  \If{$\xi_{a_j}=1$ and $a_j \not\in \cH$}
  \State $\cH\gets\cH\cup \{a_j\}$
  \While{there is $i\in \cH$ such that all $s\in \cS$ with $i\in s$ satisfy $|s\cap \cH| > U$}
  \State remove this $i$ from $\cH$
  \Comment{remove $i$ that is not $U$-needed}
  \EndWhile
  \EndIf	
  
  \State \textbf{Query():}  
  \State \quad \Return $\cH$
  \Comment{implicit answer $\supp(\xi\circ s\circ v)$ for each $s\in \cS$ }
\end{algorithmic}
\end{algorithm}


\begin{lemma}
\label{lemma:alg-sampler}
Consider Algorithm~\ref{alg:single-estimator} 
for $\cS\subset\zo^n$ with parameters $p\in (0,1]$ and $U\in [1, n]$.
When run on an insertion-only stream accumulating to $v\in \RR^n$, 
it makes one pass, uses $O(U)\cdot \HHdim(\cS)$ words of space,
and outputs $\cH\subset[n]$ of size $|\cH|\le U\cdot \HHdim(\cS)$. 
Moreover, suppose that $\xi\in\zo^n$ is the sampling vector from the algorithm.
Then for every $s\in \cS$, 
if 
$\norm{\xi\circ s\circ v}_0 \leq U$
then 
\[
   \supp(\xi\circ s\circ v) \subset \cH;
\]
and otherwise $|\cH\cap s|\ge U$.
\end{lemma}





\begin{proof}
Observe that after each update operation, 
every $i\in \cH$ is $U$-needed, i.e., 
\[
  \exists s\in \cS, \quad 
  i\in s \text{ and } |s\cap \cH| \le U.
\]
By applying Proposition~\ref{prop:subset} to this $\cH$
we have that $|\cH|\le U\cdot \HHdim(\cS)$. 
This bound applies in particular to the  output $\cH$,
and also implies that the algorithm uses $O(U)\cdot \HHdim(\cS)$ words of space. 

Now consider $\cH$ at the end of the stream, and some $s\in \cS$. 
Let $Z = \supp(\xi\circ s\circ v)$.
If $|Z|\le U$, 
then every $i\in Z$ has to been stored in the final $\cH$
(because it must be added at some update and can never be removed),
which proves that $Z \subset \cH$. 

Next, suppose that $|Z|>U$.
Observe that every index $i\in Z$ is added at some point to $\cH$, 
hence $|s\cap \cH|$ is increased more than $U$ times.
Moreover, whenever any $i\in s\cap \cH$ is removed from $\cH$, 
it may only decrease $|s\cap \cH|$ from $U+1$ to $U$, but never below $U$. 
It follows that at the end of the execution, $|s\cap \cH| \ge U$. 
\end{proof}

We shall use Algorithm~\ref{alg:single-estimator} as a subroutine twice,
first to compute an $O(1)$-approximation to a query $s\in \cS$, 
and then (a refined version of it) 
to compute a $(1\pm\epsilon)$-approximation.
En route to an $O(1)$-approximation, 
we introduce Algorithm~\ref{alg:const-detector}, 
which simply runs Algorithm~\ref{alg:single-estimator} 
in parallel $\Theta(\log \log n)$ times,
and then when given a query $s\in \cS$, it outputs one bit. 
The next lemma shows that this algorithm solves a promise (gap) version:
if $\norm{v\circ s}_0\le 1/(2p)$ then with high probability it outputs $0$, 
and if $\norm{v\circ s}_0\ge 2/p$ then with high probability it outputs $1$.

\begin{algorithm}
\caption{Constant-Detector for $\cS\subset\zo^n$ \label{alg:const-detector}}
\begin{algorithmic}[1]
  \State \textbf{Input:} $r\in [\tfrac14, n]$, an insertion-only stream $\langle a_1, \ldots, a_m\rangle$, where each item $a_j\in [n]$ 
  \State \textbf{Initialize:} 
  \State \quad $U \gets 100$, $p\gets \min(1, U/r)$, $t\gets\Theta(\log\log n)$ 
  \State \quad let $\cA_1,\ldots, \cA_t$ be instances of Bounded-Sampler with parameters $p$ and $U$
  \State\textbf{Update($a_j$):}
  \quad \State \quad update each $\cA_i$ with $a_j$ 	
  \State \textbf{Query($s\in \cS$):}   
  \State \quad query each $\cA_i$ for $s$ and let $\cH_i$ be its output 
  \State \quad $\bar z\gets \tfrac{1}{p} \cdot \operatorname{median}\big(|\cH_1\cap s|,\ldots,|\cH_t\cap s|\big)$
  \State \quad \Return $\indice{\bar z\ge r}$
\end{algorithmic}
\end{algorithm}


\begin{lemma}
\label{lemma:const-detector}
Consider Algorithm~\ref{alg:const-detector} 
for $\cS\subset\zo^n$ with parameter $r\in [\tfrac14, n]$.
When run on an insertion-only stream accumulating to $v\in \RR^n$, 
it makes one pass and uses $O(\log\log n\cdot \HHdim(\cS))$ words of space,
and then when queried for $s\in \cS$, 
with probability at least $1-O(1/\log^2 n)$ its output satisfies:
if $\|v\circ s\|_0\le r/2$ the output is $0$, 
and if $\|v\circ s\|_0 \ge 2r$ the output is $1$.
\end{lemma}

\begin{proof}
The algorithm's space usage 
is dominated by the $t$ instances of the Bounded-Sampler,
and thus follows from Lemma~\ref{lemma:alg-sampler}.

If $r\le U$, then $p=1$ and the algorithm is deterministic, 
with the following guarantee by Lemma~\ref{lemma:alg-sampler}. 
If $\norm{s\circ v}_0 \le r/2 \le U$, 
then $z=\norm{s\circ v}_0$ is an exact estimator, and the output is $0$. 
If $\norm{s\circ v}_0 \ge 2r$, 
then $z \ge \min(\norm{s\circ v}_0,U)\ge r$ and the output is $1$. 

We thus assume henceforth that $r>U$ and thus $p=U/r$. 
Consider an instance $\cA_i$ of the Bounded-Sampler,
let $\cH_i$ be its output, and let $\xi^i\in\zo^n$ be its sampling vector. 
We would like to analyze the quantity $\tfrac{1}{p} \card{\cH_i\cap s}$ 
used in the algorithm. 


Now suppose $ \|s\circ v\|_0 \ge 2r =\tfrac{2U}{p}$.
Then by Lemma~\ref{lem:qv-pairwise}, the expectation is $\ex[\tfrac{1}{p}\norm{\xi^i\circ s\circ v}_0] = \norm{s\circ v}_0$,
and with probability at least $8/9$, 
\[
 \tfrac{1}{p}\norm{\xi^i\circ s\circ v}_0 
 \ge \norm{s\circ v}_0  - 3 \big(\tfrac{1-p}{p} \|s\circ v\|_0 \big)^{1/2} 
 \ge \norm{s\circ v}_0  - \tfrac{3}{(2U)^{1/2}} \|s\circ v\|_0 
 \ge \tfrac34 \norm{s\circ v}_0
 \ge \tfrac{3U}{2p}. 
\]
In this event, 
$\norm{\xi^i\circ s\circ v}_0 \ge \tfrac{3U}{2} > U$,
which by Lemma~\ref{lemma:alg-sampler} implies that 
$\card{\cH_i\cap s} \ge U$,
and therefore the estimate obtained from the instance $\cA_i$ is
$\tfrac{1}{p} |\cH_i\cap s| \ge \tfrac{1}{p} U\ge r$.
By a standard probability amplification argument, i.e., Chernoff bound, 
with probability at least $1-O(1/\log^2 n)$, 
the median of $t$ independent repetitions is $\bar z \ge r$, 
and the output is $1$.

Suppose next that $ \|s\circ v\|_0 \le \tfrac{r}{2} = \tfrac{U}{2p}$.
Then by Lemma~\ref{lem:qv-pairwise}, 
with probability at least $8/9$, 
\[
 \tfrac{1}{p}\norm{\xi^i\circ s\circ v}_0 
 \le \norm{s\circ v}_0  + 3 \big(\tfrac{1-p}{p} \|s\circ v\|_0 \big)^{1/2} 
 \le \tfrac{U}{2p} + \tfrac{3}{p} (\tfrac{U}{2})^{1/2} 
 \le \tfrac{3U}{4p}. 
\] 
In this event, 
$ \norm{\xi^i\circ s\circ v}_0 \le \tfrac{3U}{4}$
which by Lemma~\ref{lemma:alg-sampler} implies that 
$\cH_i\cap s = \supp(\xi^i\circ s\circ v)$, 
and therefore the estimate obtained from the instance $\cA_i$ is
$\tfrac{1}{p} |\cH_i\cap s| 
  = \tfrac{1}{p} \norm{\xi^i\circ s\circ v}_0
  \le \tfrac{3U}{4p}
  < r$.
By a standard probability amplification argument, i.e., Chernoff bound, 
with probability at least $1-O(1/\log^2 n)$, 
the median of $t$ independent repetitions is $\bar z < r$, 
and the output is $0$.
\end{proof}

Using Algorithm~\ref{alg:const-detector} we can now easily design 
an algorithm achieving $O(1)$-approximation. 

\begin{lemma}[$8$-approximation algorithm]
\label{lemma:const approx}
There is an algorithm that when run
for $\cS\subset\zo^n$ with parameter $r\in [\tfrac14, n]$
on an insertion-only stream accumulating to $v\in \RR^n$, 
makes one pass and uses $O(\log n\cdot \log\log n\cdot \HHdim(\cS))$ words of space,
and then when queried for $s\in \cS$, 
it output a number $z$ that with probability at least $0.99$ satisfies 
$\|s\circ v\|_0 < z < 8 \|s\circ v\|_0$.
\end{lemma}

\begin{proof}
The algorithm consists of $l=O(\log n)$ parallel independent instances
of Algorithm~\ref{alg:const-detector}, denoted $\cB_0, \cB_1,\ldots, \cB_{l}$. 
For each instance $\cB_j$, the parameter is $r_j = 2^{j-2}$.
To process a query $s\in \cS$, the algorithm queries every instance for this $s$.
Let $x_j$ denote the output from instance $\cB_j$. 
By Lemma~\ref{lemma:const-detector}, 
$x_0 = 0$ if and only if $\|s\circ v\|_0=0$.
Hence, $x_0$ can be used to distinguish whether $\|s\circ v\|_0=0$,
and we may assume henceforth that $\|s\circ v\|_0>0$.

The algorithm computes $j^*$ which is the smallest $j\ge 1$ such that $x_j = 0$,
and outputs $z=2^{j^*}$.
By Lemma~\ref{lemma:const-detector} and a union bound, 
with probability at least $0.99$, for all $j=1,\ldots,l$, 
if $\|s\circ v\|_0 \ge 2r_j = 2^{j}$ then $x_j =1$, and 
if $\|s\circ v\|_0 \le r_{j}/2 = 2^{j-2}$ then $x_{j} =0$. 
Assuming this event happens, 
and by applying the first condition to $j^*$ and the second one to $j^*-1$ 
(both in the contrapositive form), we obtain
$ 2^{(j^*-1)-2} < \|s\circ v\|_0 < 2^{j^*} = z$.

It is easy to verify the algorithm's space requirement,
and this completes the proof.
\end{proof}

We now design an $(1\pm\epsilon)$-approximation algorithm, 
by using the above $O(1)$-approximation algorithm 
and a variant of Algorithm~\ref{alg:single-estimator}. 

\begin{algorithm}
\caption{Multiplicative Approximation\label{alg:eps-approx} for $\cS\subset \zo^n$}
\begin{algorithmic}[1]
  \State \textbf{Input:} $\epsilon\in (0,1)$, an insertion-only stream $\langle a_1, \ldots, a_m\rangle$, where each item $a_j\in [n]$ 
  \State \textbf{Initialize:} 
  \State  \quad $t\gets \Theta(\log n)$
  \State  \quad let $\set{\cA_i}_{i=1,\ldots,t}$ be instances of Bounded-Sampler with parameters $p_i = 2^{1-i}$ and $U =\lceil 400\epsilon^{-2}\rceil$
  \State  \quad let $\cB$ be an $8$-approximation algorithm (from Lemma~\ref{lemma:const approx})
  \State\textbf{Update($a_j$):}
  \State \quad Update $\cB$ and each $\cA_i$ with $a_j$
  \State \textbf{Query($s\in \cS$):}   
  \State \quad query $\cB$ for $s$ and let $z$ be its output 
  \State \quad $l\gets \max(1, \lceil\log(z\epsilon^{2}/100)\rceil )$;
  \State \quad query $\cA_l$ for $s$ and let $\cH_l$ be its output 
  \State \quad \Return $\tfrac{1}{p_l} |\cH_l\cap s|$;
\end{algorithmic}
\end{algorithm}

\begin{theorem}
\label{thm:alg-main}
Consider Algorithm~\ref{alg:eps-approx}
for $\cS\subset\zo^n$ with parameter $\epsilon\in (0,1)$.
When run on an insertion-only stream accumulating to $v\in \RR^n$, 
it makes one pass and uses 
$O((\epsilon^{-2} + \log\log n)\log n \cdot \HHdim(\cS))$ words of space, 
and then when queried for $s\in \cS$, its output $\wh{z}(s)$ satisfies
\[
  \forall s\in\cS, \qquad
  \Pr\Big[\wh{z}(s) \in (1\pm\epsilon) \|s\circ v\|_0 \Big] 
  \ge 0.8.
\] 
\end{theorem}

\begin{proof}
We assume $\|s\circ v\|_0>0$, since otherwise the algorithm specified in Lemma~\ref{lemma:const approx} already gives the correct answer.
Next, by Lemma~\ref{lemma:const approx}, with probability at least $0.99$, 
the value $z$ reported by $\cB$ is an $8$-approximation to $\|s\circ v\|_0$,
and for the rest of the proof we assume this event happens. 
Consider $l$ as in the algorithm, 
and its corresponding $p_l = 2^{1-l} = \min(1,200\epsilon^{-2}/z)$, 
then
\[
  \min(1, 25\epsilon^{-2}/\|s\circ v\|_0)
  \le p_l 
  \le \min(1, 200\epsilon^{-2}/\|s\circ v\|_0).
\]
Let $Z_l\subset \supp(s\circ v)$ be the set of indices sampled 
in algorithm $\cA_l$ by the hash function $\xi$, 
and recall it samples each index in $\supp(s\circ v)$ 
with probability $p_l$ pairwise independently.
By Lemma~\ref{lem:qv-pairwise}, with probability at least $0.8$, 
\[
  \left| \tfrac{1}{p_l} |Z_l| - \norm{s\circ v}_0 \right|
  \le 3 (\tfrac{1-p_l}{p_l} \|s\circ v\|_0)^{1/2} 
  \le \epsilon \|s\circ v\|_0 ,
\]
which implies that 
$
  |Z_l| \le (1 + \epsilon) p_l \|s\circ v\|_0 \le  400\epsilon^{-2} = U 
$. 
When this happens, by Lemma~\ref{lemma:alg-sampler} 
instance $\cA_l$ of Algorithm~\ref{alg:single-estimator} 
outputs $\cH_l$ that satisfies $\cH_l\cap s = Z_l$,
and we conclude that our algorithm's output is
$\tfrac{1}{p_l}|\cH_l\cap s| 
  = \tfrac{1}{p_l} |Z_l| 
  \in (1\pm\epsilon) \|s\circ v\|_0
$.

The proof of Theorem~\ref{thm:alg-main} is completed  
by easily verifying the space usage of the algorithm. 
\end{proof}

\subsection{Matching Lower Bound}
\label{sec:SubsetLpLB}

We prove that for every set system $\cS\subset\zo^n$, 
the space complexity of every universal streaming algorithm 
must be $\Omega(\HHdim(\cS))$, 
which matches Theorem~\ref{thm:alg-main} 
in terms of the linear dependence on the heavy-hitter dimension.

\begin{theorem}
\label{thm:lower-bound-main}
Let $\cS\subset \{0, 1\}^n$ be a non-empty set system. 
Suppose $\cA$ is a (randomized) one-pass streaming algorithm 
that solves the Subset-$\ell_p$ problem for $\cS$
within approximation factor $\alpha \ge 1$ for some $p\ge 0$. 
Then $\cA$ requires $\Omega(\HHdim(\cS))$ bits of space 
for some insertion-only stream input.
Moreover, if $\alpha = 1+\epsilon$ 
for some $\epsilon \geq 1/\sqrt{\max_{s\in \cS} \|s\|_0 }$ and $p\neq 1$, 
then $\cA$ requires $\Omega(\HHdim(\cS)+\epsilon^{-2})$ bits of space.  
\end{theorem}

\begin{proof}
We begin with the lower bound for $\Omega(\HHdim(\cS))$.
Let $k = \HHdim(\cS)\ge 1$, then there exists a vector $v\in \RR^{n}$ such that 
$|H(S,v)| \ge k$. 
Without loss of generality, we can assume $v\in \{0,1\}^n$ since replacing each non-zero coordinate of $v$ with $1$ does not change $D(s,v)$ for any $s\in \{0,1\}^n$. 

Since $\cS$ is given to the algorithm before the streaming coming, we can use the vector $v$ and the algorithm $\cS$ to design a one-way communication protocol that solves INDEX$(k)$, in which Alice is holding a binary vector $x\in \{0,1\}^k$ of dimension $k$ and Bob is holding index $i\in[k]$.
Alice needs to send one-round of message to Bob. Bob needs to figure out the $i$-th coordinate in Alice's string. 
It is well-known that any protocol with at least constant probability of success requires $\Omega(k)$ bits (e.g. \cite{knr99}).

We now describe the protocol for the INDEX problem using the algorithm for the subset-$\ell_p$ problem.
Firstly, since $|H(S, v)|\ge k$, there exists $s_1, s_2, \ldots, s_k \in \cS$ such that each $\|s_j\circ v\|_0 = 1$
and  $D(s_j, v)$s are disjoint.
Therefore, each $s_j$ uniquely picks up a coordinate in $v$.
We denote the non-zero coordinate of $s_j\circ v$ as $z_j$.
Alice and Bob (without communication) then map the $j$-th element in $[k]$ to the $z_j$-th coordinate in $v$.
Alice then modifies $v$ such that $v_{z_j} = x_j$ for each $j$.
As such, Alice obtains a vector $v'$. 
She then converts the vector $v'$ to a insertion-only stream and runs the algorithm $\cA$ on vector $v$ and sends the memory content to Bob. 
Bob recovers the instance of the algorithm $\cA$ and runs a query on approximating the $\ell_p$ norm of vector $v'\circ s_i$.
Since $\|v'\circ s_i\|_p = x_i$, Bob can recover a answer for the INDEX problem from any $\alpha$-multiplicative approximation of $\|v'\circ s_i\|_p$. 
Thus algorithm $\cA$ must use $\Omega(k)$ bits of space in the worst case.

Lastly, the $\Omega(\epsilon^{-2})$  lower bound follows from the standard lower bound for $\ell_p$ norm \cite{iw03}.
\end{proof}

\begin{remark}
For every $p\le 2$, the lower bound $\Omega(\HHdim(\cS) + \epsilon^{-2})$ 
is existentially tight up to $\polylog(n)$-factor.
Indeed, the system $\cS^*=\{e_1,\ldots,e_{k-1}, [k,n]\}$
has $\HHdim(\cS^*) =\card{\cS^*}= k$,
but admits an algorithm with space usage $O(\HHdim(\cS^*) + 1/\epsilon^2)$ 
by explicitly storing the first $k-1$ coordinates 
and running an $\ell_p$-norm algorithm for the coordinates-subset $[k,n]$.
Thus, Theorem~\ref{thm:lower-bound-main} 
provides the best-possible dependence on $\epsilon$.
\end{remark}

\begin{remark}
Some set systems $\cS\subset\zo^n$ admit a stronger lower bound hand
of $\Omega(\HHdim(\cS)\cdot \epsilon^{-2})$.
Indeed, consider a set system with $\HHdim(\cS)$ disjoint subsets
where each subset has cardinality $\epsilon^{-2}$,
then a lower bound follows by a reduction from $\HHdim(\cS)$ 
independent instances of Gap Hamming Distance, 
and apply \cite{ACKQWZ16}, which is based on \cite{CR12,BGPW13}.
Thus, the space complexity in Theorem~\ref{thm:alg-main} 
provides the best-possible dependence on $\epsilon$.
\end{remark}

\subsection{Strong Lower Bounds for Turnstile and Sliding-Window Models}
\label{sec:SubsetLpTurnstile}

We now show an impossibility result for the Subset-$\ell_p$ problem 
in richer data streams, namely, the strict turnstile and sliding-window models.
Specifically, we exhibit a family of subsets that has a small heavy-hitter dimension but does not admit efficient (nontrivial) streaming algorithms 
in those richer data streams. 
This shows a strong separation from the insertion-only model.

Recall that in the turnstile model, the stream contains additive updates
to a vector $v\in\RR^n$, which is initialized to all-zeros. 
As the updates may be negative, it captures both insertions and deletions. 
In the strict turnstile model, the coordinates of $v$ must remain
non-negative at all times. 
For an even $n$, 
let $\cSint\subset 2^{[n]}$ be the family of all intervals of length $n/2$ 
(and thus cardinality $n/2+1$), i.e., 
\[
  \cSint \eqdef \big\{ [a, a+n/2]:\ a=1,\ldots,n/2 \big\}.
\]
We next show that $\cSint$ has a small heavy-hitter dimension
(much smaller than its cardinality $|\cSint| = n/2$), 
and thus admits efficient algorithms for insertion-only streams.

\begin{proposition}
$h(\cSint)\le 3$.
\end{proposition}
\begin{proof}
For any vector $v\in \RR^n$, if for some $j\in [n]$ and $s\in \cSint$ we have $s\circ v = (0,0,\ldots,0,v_j,0, \ldots, 0)$, then there are at least $(n/2-1)$ $0$s come around coordinate $j$ in $v$.
Therefore, the total number singletons in $v$ can be at most $n/(n/2-1) \le 3$.
\end{proof}

We now show a strong space lower bound for every streaming algorithm 
in the turnstile model. 

\begin{theorem}
\label{thm:turnstile-lower-bound}
Suppose $\cA$ is a (randomized) one-pass streaming algorithm 
that solves the Subset-$\ell_p$ problem for $\cSint$
within approximation factor $\alpha \ge 1$ for some $p\ge 0$.
Then for some turnstile stream input, 
$\cA$ requires $\Omega(n)$ bits of space. 
\end{theorem}

Before proving the theorem, we recall a known communication lower bound,
which was introduced and proved in \cite{BJKK04b},
following the classical Index problem from \cite{m98}.
In the \emph{Augmented Index} problem, denoted $\mathrm{AUG}_n$, 
Alice holds a binary vector $x\in \{0,1\}^n$,
and Bob holds an index $j\in[n]$ 
and a sequence $x_{j+1},\ldots, x_n\in\zo$ (part of Alice's vector). 
Alice then sends a single round of message to Bob,
who is required to output $x_j$. 
  
\begin{theorem}[Lower Bound for Augmented Index \cite{BJKK04b}]
In every shared-randomness protocol for $\mathrm{AUG}_n$
(with success probability at least $0.9$),
Alice must send $\Omega(n)$ bits.
\end{theorem}

We are now ready to prove our theorem.

\begin{proof}[Proof of Theorem~\ref{thm:turnstile-lower-bound}]
Suppose we have an algorithm $\cA$ that solves, with $\alpha$-approximation,
the Subset-$\ell_p$ problem of $\cSint$.
We now describe a protocol that solves the $\mathrm{AUG}_{n/2}$ problem.
Now Alice has a binary vector $x\in\{0,1\}^{n/2}$, and Bob has an index $j$ and $x_{j+1}, x_{j+2}, \ldots x_{n/2}$.
Alice treats her vector $x$ as a stream of updates and feeds it to $\cA$. 
She then sends the memory content of $\cA$ to Bob.
Bob continues running $\cA$ based on the memory content received from Alice. 
He then send the $-x_{j+1}, -x_{j+2}, \ldots, -x_{n/2}$ to the stream. 
After this, Bob queries the set $s_j = [j, j+n/2]$. 
If the algorithm answers a number $>0$, Bob then claims $x_j=1$. Otherwise he claims $x_j=0$.

To show the correctness, we observe that after Bob's updates, the vector in the stream is exactly $x' = (x_1, x_2, \ldots, x_j, 0, 0, \ldots, 0)$.
Therefore, $x'\circ s_j = (0,\ldots, 0, x_j, 0, \ldots, 0)$.
Hence if $x_j = 0$, then with probability at least $0.9$, $\cA$ outputs $0$ and if  $x_j \neq 0$, then with probability at least $0.9$, $\cA$ outputs $>0$.
Thus, the lower bound of $\mathrm{AUG}_{n/2}$ implies a space lower bound of $\cA$. 
\end{proof}
 
A similar argument applies to the sliding-window model,
where the stream has a parameter $W\ge 1$ called window-size,
and the input vector $v\in\RR^n$ at any time $t$ 
is determined by the last $W$ additive updates, 
i.e., items from time $t-W$ or earlier in the stream expire (are ignored).

\begin{theorem}
\label{thm:sliding-lower-bound}
Suppose $\cA$ is a (randomized) one-pass streaming algorithm 
that solves the Subset-$\ell_p$ problem for $\cSint$
within approximation factor $\alpha \ge 1$ for some $p\ge 0$.
Then for some sliding-window stream input,
$\cA$ requires $\Omega(\min(n, W))$ bits of space. 
\end{theorem}

\begin{proof} 
Suppose we have an algorithm $\cA$ that solves, with $\alpha$-approximation,
the Subset-$\ell_p$ problem for $\cSint$ for the most recent $W$ updates at any time $t$. 
Let $d = \min(n/2, W)$.
We now describe a protocol that solves the $\mathrm{AUG}_{d}$ problem.
Now Alice has a binary vector $x\in\{0,1\}^{d}$, and Bob has an index $j$ and $x_{j+1}, x_{j+2}, \ldots x_{n/2}$.
Alice treats her vector $x$ as a stream of updates and feeds it to $\cA$ in an order $x_{d}, x_{d-1}, \ldots, x_{1}$. 
She then sends the memory content of $\cA$ to Bob.
Bob continues running $\cA$ based on the memory content received from Alice. 
He then sends many updates to $x_{j-1}$ such that all updates of $x_{d}, x_{d-1}, \ldots, x_{j+1}$ expire except the update of $x_{j}$, i.e., Bob sends $W-\sum_{j'=j+1}^dx_{j'}$ updates to the $j-1$st coordinate of $v$. 
After this, Bob queries the set $s_j = [j, j+n/2]$. 
If the algorithm answers a number $>0$, Bob then claims $x_j=1$. Otherwise he claims $x_j=0$.

It is easy to verify the correctness of the protocol and hence proves a $\Omega(\min(n, W))$ space lower bound of the algorithm.
\end{proof}

\subsection{Streaming Algorithm with ``For All'' Guarantee}
\label{sec:SubsetLpForAll}

We now show how to extend our algorithm to achieve the ``for all'' guarantee
using space usage $\wt{O}_\epsilon(\HHdim(\cS)^2)$.
The key is to establish a connection between the heavy-hitter dimension 
and the VC-dimension (of every set system $\cS$),
and the algorithm then follows by the standard technique 
of amplifying the success probability by independent repetitions.

Recall that the VC-dimension of $\cS$ is defined as
the maximum cardinality of a set $A\subset [n]$ that is shattered,
where $A$ is called shattered
if every subset of $A$ can be realized as $A\cap S$ for some $S\in \cS$.
The heavy-hitter dimension can be defined analogously,
by modifying the definition of being shattered to this:
for every element $a\in A$, there is $S\in \cS$ such that $S\cap A = \{a\}$.
It then follows easily that $\VCdim(\cS) \le \HHdim(\cS)$,
proved formally in Proposition~\ref{prop:VCdim} below. 
However the gap between them cannot be bounded by any fixed factor.
For instance, when $\cS$ is the set of $k\in[n]$ singleton sets,
$\VCdim(\cS)=1$ whereas $\HHdim(\cS)=k$.

\begin{proposition}
  \label{prop:VCdim}
Let $\cS\subset 2^{[n]}$, 
and denote its VC-dimension by $\VCdim(\cS)$. 
Then 
\[
  \VCdim(\cS) \le \HHdim(\cS).
\]
\end{proposition}
\begin{proof}
We show that $\cS$ cannot shatter any set of size $\HHdim(\cS) +1$.
Suppose $\cS$ can shatter a set $\cS$ with $|S| = \HHdim(\cS) +1$.
Then for each $a\in \cS$, we have a set $s_a\in \cS$ such that
$\{a\} = s_a\cap S$. 
This in turn indicates $\HHdim(\cS)\ge  \HHdim(\cS) +1$, a contradiction.
\end{proof}

\begin{lemma}[Sauer-Shelah Lemma~\cite{s72a,s72b}]
Every $\cS\subset 2^{[n]}$ with VC-dimension $k$
has cardinality $\card{\cS} = O(n^k)$.
\end{lemma}

The next theorem achieves the ``for all'' guarantee by standard amplification, 
namely, by reporting the median value 
$O(\log\card{\cS}) = O(\VCdim(\cS)\cdot \log n)$ independent repetitions, 
whose error probability is analyzed by a Chernoff bound 
and a union bound over all $s\in\cS$.

\begin{theorem}
For every $\cS\subset \zo^{n}$ there is a one-pass streaming algorithm, 
that when run and $\epsilon\in (0,1)$ 
on an insertion-only stream accumulating to $v\in \RR^n$, 
it uses ${O}((\epsilon^{-2}+\log\log n)\cdot \HHdim(\cS)^2\cdot\log^2 n)$ 
words of space,
and then when queried for $s\in \cS$, its output $\wh{z}(s)$ satisfies 
\[
  \Pr\Big[\forall s\in \cS,\ \wh{z}(s) \in (1\pm\epsilon) \|s\circ v\|_0 \Big] 
  \ge 0.9 .
\]
\end{theorem}

\subsection{Generalizing the Algorithm to Subset-$\ell_1$}
\label{sec:SubsetLpL1}

In the insertion-only streaming model, 
the $\ell_1$ norm of a vector $v$ is simply the sum of all updates
(effectively without absolute values). 
Thus, we can then reduce $\ell_1$ to $\ell_0$,
and obtain an algorithm for Subset-$\ell_1$, as follows.
Assuming that the stream length $m$ is bounded by $\wb{m}=\poly(n)$,
we convert each update $a_j\in [n]$ to an update of the form $(a_j, j)$
to $v'\in \RR^{n\times \wb{m}}$, a binary vector in a larger dimension.
It is easy to verify that $\|v'\|_0 = \|v\|_1$.
We also convert the set system $\cS\in2^{[n]}$ to the new universe as follows.
For each $s\in\cS$, we expand each of its entries $s_i$ to $\wb{m}$ duplicates of $s_i$, which yields a new set system $\cS'\subset 2^{[ n\times \wb{m} ]}$.  
The following lemma shows that the new set system $\cS'$ 
has the same-heavy hitter dimension as $\cS$.
\begin{lemma}
$\HHdim(\cS') = \HHdim(\cS)$.
\end{lemma}
\begin{proof}
We first show $\HHdim(\cS)\le \HHdim(\cS')$. 
For each vector $v\in \RR^n$, we can pad each coordinate with $\wb{m}-1$ zeros to obtain a vector $v'$.
Therefore, $h(v, \cS) = h(v', \cS')$ and thus $\HHdim(\cS)\le h(\cS')$.
We now show the other direction, $\HHdim(\cS')\le \HHdim(\cS)$.
For any vector $v'\in \RR^{n\times \wb{m}}$, suppose an $s'\in \cS'$ satisfies $\supp(s'\circ v') = \{(i,j)\}$ for some $(i,j)\in [n]\times[\wb{m}]$.
Then it must be the case that $v_{(i,j')}=0$ for any $j'\neq j$ since $(i,j')\in s'$ for all $j'\in [\wb{m}]$.
Suppose $h(v', S')= k$, then there exists $k$ distinct indicies $i_1, i_2, \ldots, i_k\in [n]$ and some $k$ indicies $j_1, j_2, \ldots j_k\in [\wb{m}]$ such that for each $l\in [k]$ there exists $s_l'\in \cS'$ with $\supp(s_l'\circ v') = \{(i_l,j_l)\}$.
Consider a corresponding vector $v\in \RR^{n}$ with only $v_{i_l}=1$ for $l=1, 2,\ldots, k$ and other places $0$.
Then we have, for each $l\in [k]$, $\supp(s_l\circ v)=\{i_l\}$, where $s_l$ is the corresponding set of $s_l'$.
Hence $\HHdim(\cS')\le \HHdim(\cS)$.
\end{proof}

We can now apply our algorithm for Subset-$\ell_0$ on $v'$,
and obtain an algorithm for Subset-$\ell_1$.
\begin{theorem}
\label{thm:alg-l1}
There is an algorithm that when run on an insertion-only stream 
that accumulates to $v\in \RR^n$ and has length $\poly(n)$, 
the algorithm makes one pass using
${O}((\epsilon^{-2} + \log\log n)\log n\cdot \HHdim(\cS))$ words of space,
and then when queried for $s\in \cS$, its output $\wh{z}(s)$ satisfies
\[
  \forall s\in\cS, \quad
  \Pr[\wh{z}(s) \in (1\pm\epsilon) \|s\circ v\|_1 ] \ge 0.9.
\] 
\end{theorem}

\subsection{The Entry-Wise Update Model}
\label{sec:SubsetLpEntrywise}

In this section, we show an algorithm that computes the subset-$\ell_p$ in the entry-wise update model.
For the entry-wise update model, the algorithm for $\ell_p$ is essentially equivalent for all $p\ge 0$.
This is because when we say an entry $v_i$ comes, we can simply raise the $p$-th power of it for free.
Therefore, we simply show an algorithm for the subset-$\ell_1$ problem and the algorithms follows automatically for all other $p$.
We will use the priority sampling algorithm for subset-$\ell_1$  in \cite{s06, adlt05, dlt07}.
Our idea is to simulate the priority sampling on each subset of the set system.
We then store samples for each set and remove duplicates.
We argue that if the set system has a small heavy hitter dimension, then the number of distinct coordinates to store is, in fact, small.
This results in an algorithm with small space.
\paragraph{The Priority Sampling Algorithm}
The algorithm consists the following steps.
Given an vector $v\in \RR$, the priority algorithm first samples an random number $u_i\in (0,1)$ for each $i$.
Then it assigns a priority $p_i = |v_i|/u_i$ to each $i$.
It keeps only $k$ items with the highest priorities.
Let $\tau$ be the priority of the $(k+1)$st largest priority.
Let $i_1, i_2, \ldots, i_k$ be the sampled items.
Then the estimate is given by $E = \sum_{j=1}^k\max(|v_{i_j}|, \tau)$. 
It has been shown in \cite{s06} the following theorem.
\begin{theorem}[\cite{s06}]
	\label{thm:priority}
$\ex(E) = \|v\|_1$ and $\var(E) \le \|v\|_1^2/(k-1)$.
\end{theorem}
Therefore, to obtain a $(1\pm \epsilon)$ multiplicative approximation ot $\|v\|_1$ (i.e., with probability $0.9$), it is suffice to set $k=\Theta(1/\epsilon^2)$.
With Proposition~\ref{prop:subset}, we are now ready to show a multiplicative approximation algorithm in the entry-wise update model.
\begin{theorem}
	\label{thm:entry-wise-model}
Let $\cS\subset 2^{[n]}$ be an arbitrary set system.
Let $p\ge 0$, $\epsilon\in (0,1)$.
In the entry-wise update model of an vector $v\in \RR^n$, there exists an algorithm that uses  ${O}[\HHdim(\cS)\cdot \epsilon^{-2}\cdot \log|S|]$ words of space and for each query $s$, it outputs a $(1\pm\epsilon)$ approximation to $\|s\circ v\|_p$ with probability at least $0.9$.
\end{theorem}
\begin{proof}
As discussed, since $\ell_p$ is essentially equivalent to $\ell_1$, it suffices to show the algorithm for $\ell_1$.
We will show an algorithm that uses space $O(\HHdim(\cS)\cdot \epsilon^{-2})$ and answers a query correctly with probability at least $0.9$.
The final algorithms follows from parallel repeating for every $s\in \cS$.
Our algorithm is as follows: for each entry $v_i$ that comes, we generate random number from $u_i\sim(0,1)$.
We compute its priority as $q_i = |v_i|/u_i$.
For each subset $s\in \cS$, we maintain a set of the top-$\Theta(1/\epsilon^2)$ priority items as well as the threshold (the priority of the $(k+1)$-th largest item).
We remove all the overlapping coordinates of different subsets. 
By Lemma~\ref{prop:subset}, we only store at most $\Theta(\HHdim(\cS)\cdot \epsilon^{-2})$ items.
We can therefore simulate the priority sampling for each $s\in \cS$.
By Theorem~\ref{thm:priority}, we can obtain an $(1\pm\epsilon)$ approximation to each $\|v\circ s\|_1$ with probability at least $0.9$.
This completes the proof.
\end{proof}

\section{Additive Error Subset-$\ell_p$}
\label{sec:additive} 

In this section we design additive-approximation algorithms for subset-norms
(in contrast to multiplicative approximation in the preceding section). 
Consider a set system $\cS\subset 2^{[n]}$. 
If $\cS$ includes the all-ones vector (i.e., the set $[n]$), 
then the space complexity of $\epsilon$-additive approximation
of subset-$\ell_p$ norm is clearly at least that of 
multiplicative $(1+\epsilon)$-approximation of $\ell_p$ norm. 
We present an algorithm that matches this lower bound, 
up to $\poly(\epsilon^{-1}\log n)$ factors, for every $p\in (0, \infty)$.
It works for all possible subsets, i.e., the set system $\cS = 2^{[n]}$. 

\begin{theorem}
\label{thm:additive}
Given $p\in [0,2)$ and $\epsilon \in (0,1)$,
Algorithm~\ref{alg:lp set sektch} makes
a single pass over a data stream of additive updates to a vector $v\in \RR^{n}$,
and outputs a function $F:\{0,1\}^n\rightarrow\RR$ that satisfies
\[
  \forall s\in \{0,1\}^n,
  \qquad 
  \PP\Big[ F(s) \in \|s\circ v\|_p \pm 2\epsilon \|v\|_p \Big]
  \ge 0.75,
\]
where the probability is over the algorithm's randomness. 
Moreover, the space complexity of the algorithm and the function $F$ is  $O_p(\epsilon^{-3}\polylog(n))$ bits for $p\le 2$ and $O_p(\epsilon^{-3}n^{1-2/p}\polylog(n))$ bits for $p>2$. 
\end{theorem}


At a high-level, Algorithm~\ref{alg:lp set sektch}
follows the framework developed in~\cite{knw10b, ako11, bvwy18},
where every coordinate is scaled at random 
and then the Count-Sketch algorithm is used to find heavy-hitters. 
We employ a specific simple method established in~\cite{bvwy18},
and thus need the following definition and lemma from their work. 

\begin{definition}[$\alpha$-inverse distribution~\cite{bvwy18}]
Let $\alpha\in(0,\infty)$.
A random variable $X\in\NN=\set{1,2,\ldots}$
has an \emph{$\alpha$-inverse distribution} if 
\begin{equation} \label{eq:InverseDist}
  \forall x\in \NN, 
  \quad
  \PP[X < x] = 1-\frac{1}{x^\alpha}.
\end{equation}
\end{definition}

\begin{lemma}[Lemma~4 in~\cite{bvwy18}]
\label{lemma:inverse p-norm}
Let $p\in(0,\infty)$ and $\epsilon \in(0,1)$.
Let $k\ge c\epsilon^{-2}$ be an even integer
for large enough $c_p>0$ that depends only on $p$. 
Let $X\in \RR^{k\times n}$ be a random matrix, 
whose entries have a $p$-inverse distribution 
and they are pairwise independent. 
Given $v\in \RR^n$, let $X\circ v \in \RR^{k\times n}$
be a matrix given by $(X\circ v)_{i,j} \eqdef X_{i,j}v_j$. 
Let $Z$ be the $(k/2)$-largest entry in absolute value in $X\circ v$. 
Then 
\[
  \PP\Big[ 2^{-1/p} |Z| \in (1\pm\epsilon) \|v\|_p \Big] \ge 0.9 . 
\]
\end{lemma}

We henceforth define $V \eqdef X\circ v \in \RR^{kn}$,
and view it as a vector by simply flattening the $k\times n$ matrix. 
Throughout, let $V_{\mathrm{tail}(k)}$ denote the vector obtained from $V$ 
by zeroing the $k$ entries of largest absolute value. 
Thus, 
\[
  \|V_{\mathrm{tail}(k)}\|_2^2 \eqdef \sum_{j=k+1}^{kn}V_{[j]}^2 ,
\]
where $U_{[j]}$ denotes the $j$-largest coordinate in absolute value in $U$.

Another ingredient is the famous-known Count-Sketch algorithm of~\cite{CCF04},
but we need a well-known and slightly stronger guarantee from~\cite{CM06},
as follows. 
\begin{proposition}[Lemma 7 in \cite{CM06}]
\label{prop:CS} 
There is a one-pass algorithm with parameters $k\in\NN$ and $\epsilon'\in(0,1)$,
that given a stream of additive updates to a vector $V\in \RR^{n'}$,
uses space of $O(k/{\epsilon'}^2\cdot\log^3 n')$ bits
to output an estimate $\wh{V}\in\RR^{n'}$,
such that with high probability $1-1/n^2$, 
\[
  \|\wh{V} - V\|_{\infty} \le \frac{\epsilon'}{\sqrt{k}} \|V_{\mathrm{tail}(k)}\|_2 .
\]
\end{proposition}

We can now present our Algorithm~\ref{alg:lp set sektch}
which is used in Theorem~\ref{thm:additive}.
The idea is to use the estimator from Lemma~\ref{lemma:inverse p-norm},
but in order to save space, 
instead of storing $V=X\circ v\in \RR^{kn}$ explicitly,
it estimates this vector using the Count-Sketch algorithm
(with parameters $k$ and $\epsilon'=\epsilon'(p,\epsilon,n)$
given in line~\ref{line:epsprime} of Algorithm~\ref{alg:lp set sektch}). 
We write $V_{|s}$, for $s\in\zo^n$ (representing $s\subset [n]$), 
to denote the vector obtained by restricting $V$ to
entries corresponding to $(i,j)$ where $i\in[k]$ and $j\in s$ 
(i.e., zeroing all other entries). 

\begin{algorithm}
\caption{Additive Subset-$\ell_p$ of $v\in\RR^n$ \label{alg:lp set sektch}} 
\begin{algorithmic}[1]
  \State \textbf{Input:} $p\in (0,\infty)$ and $\epsilon\in(0,1)$ 
  \State \textbf{Initialize:} 
  \State $k\gets \Theta(\epsilon^{-2})$ 
  \Comment{$k$ is an even integer}

  \State generate a random matrix $X\in\RR^{k\times n}$
  whose entries have a $p$-inverse distribution
  and they are pairwise independent 
  \State initialize a Count-Sketch instance \texttt{CS}
  for a vector $V\in\RR^{kn}$ with parameters $k$ and 
  \begin{align*}
    \epsilon' = c'_p\cdot 
    \begin{cases} 
      \epsilon & \text{for $p<2$}; \\
      \epsilon/(\log n)^{1/2} & \text{for $p=2$}; \\
      \epsilon/n^{1/2-1/p} & \text{for $p>2$} 
    \end{cases}
  \end{align*}
  for a suitable constant $c'_p>0$ that depends on $p$
  \label{line:epsprime}

  \State \textbf{Update($i, \Delta$):} 
  \State feed the Count-Sketch instance \texttt{CS} with $k$ updates:
  \Comment{maintain $V=X\circ v$} 
  \[
    \big((1, i), X_{1,i}\Delta\big),
    \big((2, i), X_{2,i}\Delta\big),
    \ldots,
    \big((k, i), X_{k,i}\Delta\big) .
  \]
  \State \textbf{Query($s\in \{0,1\}^n$):}
  \State let $\wh{V}$ be the estimate of $V\in\RR^{kn}$ provided by \texttt{CS}
  \State let $\wh{z}$ be the $(k/2)$-largest coordinate in absolute value in $\wh{V}_{|_s}$.
  \State 
  \Return $2^{-1/p} |\wh{z}|$.
\end{algorithmic}
\end{algorithm}

Before proving Theorem~\ref{thm:additive}, we need the next lemma
to bound the error of the Count-Sketch algorithm in our setting. 
Its proof follows a direct calculation and appears in Section~\ref{sec:tail}.
\begin{lemma}
\label{lemma:tail}
Let $V \eqdef X\circ v\in \RR^{kn}$ be a vector 
defined as in Lemma~\ref{lemma:inverse p-norm}.
Then with probability at least $0.9$,
\[
  \big\|V_{\mathrm{tail}(k)}\big\|_2^2
  \le C_p\cdot\begin{cases}
    {k\|v\|_p^{2}} & \text{if $p<2$}, \\
    {k\|v\|_2^2} \le k \|v\|_p^2\cdot n^{1-2/p}  & \text{if $p>2$}, \\
    {k\|v\|_2^2 \cdot \log n} & \text{if $p=2$} .
  \end{cases}
\]
holds for a suitable $C_p>0$ that depends only on $p$.
\end{lemma}

\begin{proof}[Proof (of Theorem~\ref{thm:additive})]
We know by Lemma~\ref{lemma:inverse p-norm} that to estimate $\|v\circ s\|_p$,
it suffices to find the $(k/2)$-largest coordinate in absolute value in $V_{|_s}$ for $k=\Theta(\epsilon^{-2})$, and report its absolute value scaled by $2^{-1/p}$.
Observe that the input to the Count-Sketch instance \texttt{CS}
is exactly the vector $V=X\circ v$,
and thus, with probability at least $0.99$,
its output vector $\wh{V}$ satisfies 
\begin{equation} \label{eq:CSbound}
  \|\wh{V} - V\|_{\infty}
  \le \frac{\epsilon'}{\sqrt{k}} \|V_{\mathrm{tail}(k)}\|_2 ,
\end{equation}
where $\epsilon'$ is defined in
line~\ref{line:epsprime} of Algorithm~\ref{alg:lp set sektch},
and we assume henceforth this event occurs. 

Now let $\wh{z}$ be as in the algorithm,
i.e., the $(k/2)$-largest coordinate in absolute value in $\wh{V}_{|_s}$,
and let $z$ be similarly in $V_{|_s}$. 
It follows easily from~\eqref{eq:CSbound} that
\begin{equation} \label{eq:zhatError}
  | \wh{z} - z | \le \frac{\epsilon'}{\sqrt{k}} \|V_{\mathrm{tail}(k)}\|_2.
\end{equation}
Indeed, by definition of $z$, 
at least $k/2$ coordinates in $V_{|_s}$ have value at least $z$,
and~\eqref{eq:CSbound} implies that 
all the corresponding coordinates in $\wh{V}_{|_s}$ 
have value at least $z - \frac{\epsilon'}{\sqrt{k}} \|V_{\mathrm{tail}(k)}\|_2$, 
which implies $\wh{z} \ge z - \frac{\epsilon'}{\sqrt{k}} \|V_{\mathrm{tail}(k)}\|_2$.
The other direction is proved similarly. 

We now bound the additive error in \eqref{eq:zhatError}
using Lemma~\ref{lemma:tail}. 
By plugging in the value of $\epsilon'$
(depending on whether $p<2$, $p>2$, or $p=2$)
and setting a sufficiently small $c'_p>0$, 
with probability at least $0.9$,
\[
  \frac{\epsilon'}{\sqrt{k}}\|V_{\mathrm{tail}(k)}\|_2
  \le c'_p C_p^{1/2} \cdot \epsilon \norm{v}_p 
  \le 2^{1/p}\cdot \epsilon\|v\|_p ,
\]
thus $2^{-1/p} | \wh{z} - z | \le \epsilon \norm{v}_p$. 
The claimed overall accuracy now follows, via a union bound,
from this last bound and Lemma~\ref{lemma:inverse p-norm}. 
 
To complete the proof, observe that the space complexity is dominated
by that of the Count-Sketch instance,
which is $O(k/\epsilon'^2 \cdot\polylog(n))$ bits, as required. 
\end{proof}

\paragraph{Boosting The Probability.}
The algorithm in Theorem~\ref{thm:additive} is oblivious to the vector $s$,
and thus one can boost its success probability
by parallel repetition and reporting the median estimate. 
In particular, with only an $O(\log |\cS|)$ factor increase in the space,
the additive approximation will hold simultaneously for all the sets $S\in\cS$.

\subsection{Proof of Lemma~\ref{lemma:tail}}
\label{sec:tail}

\begin{proof}[Proof of Lemma~\ref{lemma:tail}]
We will prove the lemma separately for $p>2$ and for $p\le 2$. 

\paragraph{Case $p> 2$:}
By definition of $V$, 
\[
  \norm{V}_2^2 = \sum_{i\in[k], j\in[n]} V_{i,j}^2
  \quad\text{ and }\quad
  \ex[\norm{V}_2^2] = \sum_{i\in[k], j\in[n]} v_{j}^2\ \ex[X_{i,j}^2]. 
\]
To calculate $\ex[X_{i,j}^2]$,
observe that~\eqref{eq:InverseDist} implies
$\Pr[X_{i,j}= x] = 1/x^p - 1/(x+1)^p$ for all $x\in\NN$, 
and therefore
\[
  \ex[X_{i,j}^2] 
  =   \sum_{x=1}^\infty \Big(\frac{x^2}{x^{p}}-\frac{x^2}{(x+1)^{p}} \Big)
  =   \sum_{x=1}^\infty \frac{x^2 }{x^{p}} - \sum_{x=2}^\infty \frac{(x-1)^2}{x^{p}}  
  \le C'_p\sum_{x=1}^\infty \frac{1}{x^{p-1}} 
  \le C''_p. 
\]
for some constants $C'_p,C''_p>0$ that depend on $p$. 
Thus, $\ex[ \norm{V}_2^2 ] \le C''_p k \norm{v}_2^2$,
and by Markov's inequality, we obtain as claimed
\[
  \PP\Big[ \norm{V}_2^2 \ge 10C''_p k \norm{v}_2^2 \Big]
  \le \frac{1}{10} ,
\] 
and we can also plug in the well-known comparison of norms
$\norm{v}_2 \leq \norm{v}_p^{1/2-1/p}$.

\paragraph{Case $p\le 2$:}
Define the set 
\[
  W \eqdef \Big\{(i,j)\in[k]\times[n]:\ |V_{i, j}| \ge {20}^{1/p}\|v\|_p \Big\}.
\]
A simple calculation shows that 
\[
  \ex[|W|]
  = k \sum_{j\in[n]} \Pr\Big[ |v_j| X_{i,j} \ge {20}^{1/p}\|v\|_{p} \Big] 
  = k \sum_{j\in[n]} \frac{|v_j|^p}{20\|v\|_p^p} 
  = \frac{k}{20}.
\]
By Markov's inequality, the event $\cE = \set{ |W| \ge k}$
has probability $\PP[\cE] \le \frac{1}{20}$. 
When the complement event $\bar \cE$ occurs, 
every coordinate of $V_{\mathrm{tail}(k)}$ is not in $W$, 
i.e., has magnitude smaller than ${20}^{1/p}{\|v\|_{p}}$. 
Let us define the random variable 
\[
  R
  \eqdef \sum_{i\in[k], j\in[n]} v_{j}^2X_{i,j}^2
  \cdot \indice{ |v_j| X_{i,j} \le {20}^{1/p} \|v\|_{p}} . 
\]
When $\bar\cE$ occurs, clearly 
$\|V_{\mathrm{tail}(k)}\|_2^2 \leq R$,
and we thus wish to bound $R$ (with high probability). 

To this end, observe that for all $i\in[k],j\in[n]$ with $v_j\neq 0$, 
\begin{align*}
  \ex\Big[ X_{i,j}^2 \cdot \indice{ |v_j| X_{i,j} \le {20}^{1/p} \|v\|_{p} } \Big]
  &=
    \sum_{x=1}^{(20)^{1/p} \|v\|_p/|v_j|} 
    \Big( {\frac{x^2}{x^{p}} - \frac{x^2}{(x+1)^{p}}} \Big)
  \\
  &\le
    C_p' \cdot \sum_{x=1}^{(20)^{1/p} \|v\|_p/|v_j|} 
    \frac{1}{x^{p-1}}
  \\
  &\le C''_p \cdot
    \begin{cases}
      (\|v\|_{p}/|v_j|)^{2-p} & \text{if $p<2$}, \\ 
      \log n & \text{if $p=2$},  
    \end{cases}
\end{align*}
for some constants $C'_p,C''_p>0$ that depend on $p$,
where we used the fact that $m=\poly(n)$ and thus $\log m = O(\log n)$. 
It immediately follows that
\begin{align*}
  \ex[R] 
  \le 
  C''_p\cdot
  \begin{cases}
    k \|v\|_p^2 & \text{if $p<2$}, \\ 
    k \|v\|_2^2 \cdot\log n & \text{if $p=2$},  
  \end{cases}
\end{align*}
By Markov's inequality, 
\begin{align*}
  & \PP\Big[R \ge 20 C''_p\cdot k \|v\|_p^2 \Big] \le \frac{1}{20}
  & \text{for $p <2$}
  \\
  \text{and}\quad
  & \PP\Big[R \ge 20 C''_p\cdot k \|v\|_2^2 \log n \Big] \le \frac{1}{20}
  & \text{for $p=2$} .
\end{align*}
Now by a union bound on the above event and $\bar\cE$,
with probability at least $0.9$ we have
\begin{align*}
  \|V_{\mathrm{tail}(k)}\|_2^2
  \le R
  < 
  \begin{cases}
    20 C''_p\cdot k \|v\|_p^2 & \text{if $p<2$}, \\ 
    20 C''_p\cdot k \|v\|_2^2 \log n  & \text{if $p=2$},  
  \end{cases}
\end{align*}
which completes the proof of this case. 
\end{proof}


\section{Concluding Remarks}
To conclude, we study the universal streaming problem for subset-$\ell_p$-norms, 
i.e., providing a single summary of a stream of \emph{insertion-only} updates
to an input vector $v\in\RR^n$,
which suffices to approximate any \emph{subset-$\ell_0$-norm} 
from a given family $\cS\subset 2^{[n]}$.
(Recall that a subset-$\ell_p$-norm of $v$ is the $\ell_p$-norm of
the vector induced by a subset of coordinates $S\in\cS$.)
We prove that the space complexity of this problem
is characterized by the \emph{heavy-hitter dimension} of the set $\cS$,
a notion that we introduce and define as the maximum number (over all $v\in\RR^n$) of distinct heavy-hitters with respect to all subsets $S\in\cS$.
We further show that this characterization holds also for subset-$\ell_1$-norms
in the same insertion-only setting.
However, it \emph{does not} hold for more general streaming models,
namely, for the \emph{turnstile setting} and the \emph{sliding-window setting},
and thus there is a strict separation between these models.

For subset-$\ell_p$-norms with general $p$,
namely, every $p\in(0, \infty)\backslash \{1\}$,
we prove that the \emph{heavy-hitter dimension} characterizes the space complexity of universal streaming in the \emph{entry-wise} updates model,
where each coordinate of the vector is updated at most once.
In the more general model of insertion-only updates,
it is remains open whether subset-$\ell_p$-norms for $p\neq 0,1$
admits uniform streaming with space complexity $\wt{O}(\HHdim(\cS))$. 
For example, the major obstacle for subset-$\ell_2$-norms
is how to maintain the distinct $\ell_2$-heavy-hitters
for every subset of coordinates $S\in\cS$.
We leave this problem for future investigations.


{\small
\bibliographystyle{alphaurlinit}
\bibliography{ref}
}


\end{document}